\documentclass[11pt,a4paper]{article}
\usepackage{amsmath}
\usepackage{amsfonts}
\usepackage{amssymb}
\usepackage{makeidx}
\usepackage{graphicx}
\usepackage[left=3cm,right=3cm,top=3cm,bottom=3cm]{geometry}
\usepackage{algorithm,algorithmic}

\usepackage{natbib} 

\newtheorem{theorem}{Theorem}

\newtheorem{definition}[theorem]{Definition}

\newtheorem{remark}[theorem]{Remark}
\newenvironment{proof}[1][Proof]{\noindent\textbf{#1.} }{\ \rule{0.5em}{0.5em}}

\begin{document}

\author{N. Balakrishnan\footnote{McMaster University, ON, Canada. email: bala@mcmaster.ca }, E. Castilla\footnote{Complutense University of Madrid, Spain. email: elecasti@ucm.es }, N. Martin\footnote{Complutense University of Madrid, Spain. email: nirian@estad.ucm.es} \ and L. Pardo\footnote{Complutense University of Madrid, Spain. email: lpardo@ucm.es }}
\title{\textbf{ Divergence-based robust inference under proportional hazards model for one-shot device life-test}}
\date{}
\maketitle

\begin{abstract}
In this paper, we develop robust estimators and tests for one-shot device testing under proportional hazards assumption based on divergence measures. Through a detailed Monte Carlo simulation study and a numerical example, the developed inferential procedures are shown to be more robust than the classical procedures, based on maximum likelihood estimators.
\end{abstract}
\section{Introduction}

A one-shot device is a unit that performs its function only once, and after use the device  either gets destroyed or must be rebuilt.  For this kind of device, one can only know whether the failure time is either before or after a specific time, and consequently the  lifetimes are either left- or right-censored, with the lifetime being less than the inspection time if the test outcome is a failure (resulting in left censoring) and the lifetime being more than the inspection time if the test outcome is a success (resulting in right censoring). Some examples of such one-shot devices include  automobile air bags, missiles (\cite{olwell2001}) and fire extinguishers (\cite{newby2008}).

For devices with long lifetimes, accelerated life-tests (ALTs) are commonly used to induce quick failures.  An ALT shortens the life span of the products by increasing the levels of stress factors, such as temperature, humidity, pressure and voltage. Then, a link function relating stress levels and lifetime is applied to extrapolate the lifetimes of units from accelerated conditions to normal operating conditions. The study of one-shot device from ALT data has been discussed considerably recently, mainly motivated by the work of \cite{fan2009_JSCS}.

Under the classical parametric setup, product lifetimes are assumed to be fully described by a probability distribution involving some model parameters. This has been done with some common lifetime distributions such as exponential (\cite{balakrishnan2012_IEEEREL}), gamma  or Weibull (\cite{balakrishnan2013_IEEEREL}). However, as data from one-shot devices do not contain actual lifetimes, parametric inferential methods can be very sensitive to violations of the model assumption.  \cite{ling2015_IEEERELIABILITY} proposed a semi-parametric model, in which, under the proportional hazards assumption,  the hazard rate is allowed to change in a non-parametric way. The simulation study carried out by \cite{ling2015_IEEERELIABILITY} shows that their proposed method works very well. However, this method suffer from lack of robustness, as it is based on the  (non-robust) maximum likelihood estimator (MLE) of model parameters. Recently years, some work has been done for developing robust methods for one-shot device testing, most of it based on divergence measures (see, for example, \cite{balakrishnan2019_IEEEIT,balakrishnan2019_METRIKA, balakrishnan2019_IEEERELIABILITY}).

In this paper, we extend the robust approach proposed in the above mentioned papers and develop here  robust estimators and tests for one-shot device testing based on divergence measures under proportional hazards model.  Section \ref{sec:PH_model} described the model and some basic concepts and results. The estimating equations and asymptotic  properties of the proposed estimators are given in Section \ref{sec:PH_esti}. Wald-type tests are then developed in Section \ref{sec:PH_wald} based on the proposed estimators, as a generalization of the classical Wald test. In Section \ref{sec:PH_sim}, a simulation study is carried out to demonstrate the robustness of the proposed method.  A numerical example is finally presented in Section \ref{sec:PH_data}, and some concluding comments are finally made in Section \ref{sec:PH_conc}.

\section{Model formulation \label{sec:PH_model} }
Consider $S$ constant-stress accelerated life-tests and $I$ inspection times. For the $i$-th life-test, $K_{s}$ devices are placed under stress level combinations with $J$ stress factors, $\boldsymbol{x}_s=(x_{s1},\dots,x_{sJ})$, of which $K_{is}$ are tested at the $i$-th inspection time $IT_i$, where $K_s=\sum_{i=1}^IK_{is}$ and $0<IT_1<\cdots<IT_I$.  Then, the numbers of devices that have failed by time $IT_i$ at stress $\boldsymbol{x}_s$ are recorded as $n_{is}$. One-shot device testing data obtained from such a life-test can then be represented as $(n_{is}, K_{is}, \boldsymbol{x}_s,IT_i)$, for $i=1,2,\dots,I$ and $s=1,2,\dots,S$.

Instead of assuming that the true lifetimes of devices follow a specific parametric distribution such as exponential, gamma or Weibull,  we  assume here that  the cumulative hazard function of the lifetimes of devices is of the proportional form
\begin{align}\label{eq:PH_cumhazard}
H(t, \boldsymbol{x};\boldsymbol{\eta},\boldsymbol{\alpha})=H_0(t;\boldsymbol{\eta})\lambda(\boldsymbol{x};\boldsymbol{\alpha}),
\end{align}
where $H_0(t;\boldsymbol{\eta})$ is the baseline cumulative hazard function with $\boldsymbol{\eta}=(\eta_1,\dots,\eta_I)$, and $\boldsymbol{\alpha}=(\alpha_1\dots,\alpha_J)$ is a vector of coefficients for stress factors. The model in (\ref{eq:PH_cumhazard}) is thus composed of two independent components, with one measuring  the changes in the baseline ($H_0(t;\boldsymbol{\eta})$) and the other influencing the stress factors ($\lambda(\boldsymbol{x};\boldsymbol{\alpha})$). 

The corresponding reliability function is given by
\begin{align}
R(t,\boldsymbol{x};\boldsymbol{\eta},\boldsymbol{\alpha})=\exp\left(- H(t, \boldsymbol{x};\boldsymbol{\eta},\boldsymbol{\alpha})\right)=R_0(t;\boldsymbol{\eta})^{\lambda(\boldsymbol{x};\boldsymbol{\alpha})},
\end{align}
where $R_0(t;\boldsymbol{\eta})=\exp(-H_0(t;\boldsymbol{\eta}))$ is the baseline reliability function, with $0<R_0(IT_I;\boldsymbol{\eta})<R_0(IT_{I-1};\boldsymbol{\eta})<\dots<R_0(IT_1;\boldsymbol{\eta})<1$. Therefore, we let
\begin{align*}
\gamma(\eta_{i}) = \left\{ \begin{array}{ll}
         1-R_0(IT_I;\boldsymbol{\eta})=1-\exp(-\exp(\eta_I)), & i=I,\\
        \dfrac{1-R_0(IT_i;\boldsymbol{\eta})}{1-R_0(IT_{i+1};\boldsymbol{\eta})}=1-\exp(-\exp(\eta_i)), & i=1,\dots,I-1.\end{array} \right. 
\end{align*}
We then have
\begin{align*}
R_0(IT_i;\boldsymbol{\eta})=1-\prod_{m=i}^I\left\{1-\exp(-\exp(\eta_m))  \right\}=1-G_i,
\end{align*}
where $G_i=\prod_{m=i}^I\left\{1-\exp(-\exp(\eta_m))\right\}$.

We now assume a log-linear link function for relating the stress levels to the failure times of the units in the cumulative hazard function  in  (\ref{eq:PH_cumhazard}), as
$$
\lambda(\boldsymbol{x}_s;\boldsymbol{\alpha})=\exp(\boldsymbol{\alpha}^T\boldsymbol{x}_s)=\exp\left(\sum_{j=1}^J\alpha_jx_{sj} \right).
$$

\subsection{Maximum likelihood estimator}

Consider the proportional hazards model for one-shot devices in (\ref{eq:PH_cumhazard}). The log-likelihood function based on these data is then given by

\begin{align} \label{eq:PH_likelihood}
\ell(n_{11},\dots,n_{IS};\boldsymbol{\eta},\boldsymbol{\alpha})=&\sum_{i=1}^I\sum_{s=1}^S n_{is}\log \left[1-R(IT_i,\boldsymbol{x}_s;\boldsymbol{\eta},\boldsymbol{\alpha}) \right] \nonumber \\ 
& + (K_{is}-n_{is})\log\left[R(IT_i,\boldsymbol{x}_s;\boldsymbol{\eta},\boldsymbol{\alpha}) \right] + C \nonumber\\
=&\sum_{i=1}^I\sum_{s=1}^S n_{is}\log \left[1-(1-G_i)^{\exp\left(\sum_{j=1}^J\alpha_jx_{sj} \right)}\right] \nonumber \\ 
& + (K_{is}-n_{is})\log\left(1-G_i \right)\exp\left(\sum_{j=1}^J\alpha_jx_{sj} \right) + C,
\end{align}
where $C$ is a constant not depending on $\boldsymbol{\eta}$ and $\boldsymbol{\alpha}$. 

\begin{definition}\label{def:PH_mle}
Let $\boldsymbol{\theta}=(\boldsymbol{\eta},\boldsymbol{\alpha})$.  The MLE, $\widehat{\boldsymbol{\theta}}$, of $\boldsymbol{\theta}$, is  obtained by maximization of (\ref{eq:PH_likelihood}), i.e.,
\begin{equation}
\widehat{\boldsymbol{\theta}}=\underset{\boldsymbol{\theta}}{\arg \min } \ \ell(n_{11},\dots,n_{IS};\boldsymbol{\eta},\boldsymbol{\alpha}).
\end{equation}
\end{definition}

In order to study the relation between the MLE, $\widehat{\boldsymbol{\theta}}$, in Definition \ref{def:PH_mle}, with the Kullback-Leibler divergence measure, we introduce  the empirical and theoretical probability vectors, as follows:
\begin{align}
\widehat{\boldsymbol{p}}_{is}&=\left(\widehat{p}_{is1},\widehat{p}_{is2} \right)^T=\left( \frac{n_{is}}{K_{is}},\frac{K_{is}-n_{is}}{K_{is}}\right)^T, \quad i=1,\dots,I,\ s=1,\dots,S, \label{eq:PH_empvector}\\
\boldsymbol{\pi}_{is}(\boldsymbol{\eta},\boldsymbol{\alpha})&=\left(\pi_{is1}(\boldsymbol{\eta},\boldsymbol{\alpha}),\pi_{is2}(\boldsymbol{\eta},\boldsymbol{\alpha}) \right)^T,  \quad i=1,\dots,I,\ s=1,\dots,S,\label{eq:PH_theovector}
\end{align}
where \ $\pi_{is1}(\boldsymbol{\eta},\boldsymbol{\alpha})=1-R(IT_i,\boldsymbol{x}_s;\boldsymbol{\eta},\boldsymbol{\alpha})$ and \ $\pi_{is2}(\boldsymbol{\eta},\boldsymbol{\alpha})=R(IT_i,\boldsymbol{x}_s;\boldsymbol{\eta},\boldsymbol{\alpha})$.

\begin{definition}
The Kullback-Leibler divergence measure between $\widehat{\boldsymbol{p}}_{is}$ and $\boldsymbol{\pi }_{is}(\boldsymbol{\eta},\boldsymbol{\alpha})$\  is given by 
\begin{align*}
d_{KL}(\widehat{\boldsymbol{p}}_{is},\boldsymbol{\pi}_{is}(\boldsymbol{\eta},\boldsymbol{\alpha})) =\widehat{p}_{is1}\log \left( \dfrac{\widehat{p}_{is1}}{\pi_{is1}(\boldsymbol{\eta},\boldsymbol{\alpha})}\right) +\widehat{p}_{is2}\log \left( \dfrac{\widehat{p}_{is2}}{\pi_{is2}(\boldsymbol{\eta},\boldsymbol{\alpha})}\right)  
\end{align*}
and similarly the weighted Kullback-Leibler divergence measure of all the units, where $K=\sum_{s=1}^S K_s$ is the total number of devices under the life-test, is given by
\begin{align}
&\sum_{i=1}^{I}\sum_{s=1}^{S}\frac{K_{is}}{K}d_{KL}(\widehat{\boldsymbol{p}}_{is},\boldsymbol{\pi}_{is}(\boldsymbol{\eta},\boldsymbol{\alpha})) \nonumber\\
&=\frac{1}{K}\sum_{i=1}^{I}\sum_{s=1}^{S}K_{is}\left[\widehat{p}_{is1}\log \left( \dfrac{\widehat{p}_{is1}}{\pi_{is1}(\boldsymbol{\eta},\boldsymbol{\alpha})}\right) + \widehat{p}_{is2}\log \left( \dfrac{\widehat{p}_{is2}}{\pi_{is1}(\boldsymbol{\eta},\boldsymbol{\alpha})}\right)\right] \nonumber \\
&=\frac{1}{K}\sum_{i=1}^{I}\sum_{s=1}^{S}\left[n_{is}\log \left( \dfrac{\frac{n_{is}}{K_{is}}}{1-R(IT_i,\boldsymbol{x}_s;\boldsymbol{\eta},\boldsymbol{\alpha})}\right) + (K_{is}-n_{is})\log \left( \dfrac{\frac{K_{is}-n_{is}}{K_{is}}}{R(IT_i,\boldsymbol{x}_s;\boldsymbol{\eta},\boldsymbol{\alpha})}\right)\right].
\end{align}
\end{definition}

For more details, one may refer to \cite{pardo2005}. The relation between the MLE and the estimator obtained by minimizing the weighted Kullback-Leibler divergence measure is obtained on the basis on the following theorem.

\begin{theorem}
\label{res:dkull} The log-likelihood function $\ell(n_{11},\dots,n_{IS};\boldsymbol{\eta},\boldsymbol{\alpha})$, given in (\ref{eq:PH_likelihood}), is related to the weighted Kullback-Leibler divergence measure through 
\begin{equation*}
\sum_{i=1}^{I}\sum_{s=1}^{S}\frac{K_{is}}{K}d_{KL}(\widehat{\boldsymbol{p}}_{is},\boldsymbol{\pi}_{is}(\boldsymbol{\eta},\boldsymbol{\alpha}))=c-\frac{1}{K}\ell(n_{11},\dots,n_{IS};\boldsymbol{\eta},\boldsymbol{\alpha}),
\end{equation*}%
with $c$ being a constant not dependent on  $\boldsymbol{\eta}$ and $\boldsymbol{\alpha}$.
\end{theorem}

\begin{definition}
 The MLE, $\widehat{\boldsymbol{\theta}}$, of $\boldsymbol{\theta}$,  can then be defined as 
\begin{equation}
\widehat{\boldsymbol{\theta}}=\underset{\boldsymbol{\theta}}{\arg \min }\sum_{i=1}^{I}\sum_{s=1}^{S}\frac{K_{is}}{K}d_{KL}(\widehat{\boldsymbol{p}}_{is},\boldsymbol{\pi}_{is}(\boldsymbol{\eta},\boldsymbol{\alpha})).  \label{eq:PH_mle_kull}
\end{equation}
\end{definition}

\begin{remark}\label{rem:PH_Weibull}
Suppose the lifetimes of one-shot devices under test follow the Weibull distribution with the same shape parameter $\tau=exp(b)$ and  scale parameters related to the stress levels, $a_s=exp(\sum_{j=1}^J c_jx_{sj})$, $s=1,\dots,S$. The cumulative distribution function  of the Weibull distribution is then given by
\begin{equation*}
F_T(t;a_s,\tau)=1-exp\left(-\left(\frac{t}{a_s}\right)^{\tau}\right), \quad t>0.
\end{equation*}
If the proportional hazards assumption holds, then the baseline reliability and the coefficients of stress factors are given by
\begin{equation*}
R_0(t;\beta)=exp(-t^{\tau}exp(-\tau c_0))
\end{equation*}
and $\alpha_s=-\tau c_s, \quad s=1,\dots,S$.
Furthermore, we have
\begin{align*}
\eta_i&=log\left(-log\left( 1-\frac{1-R_0(IT_i)}{1-R_0(IT_{i+1})}\right)\right),\\
\eta_I&=\tau (log(IT_I)-c_0).
\end{align*}
\end{remark}

\subsection{Weighted minimum DPD estimator}

Given the probability vectors $\widehat{\boldsymbol{p}}_{is}$ and $\boldsymbol{\pi }_{is}(\boldsymbol{\eta},\boldsymbol{\alpha})$  in (\ref{eq:PH_empvector}) and  (\ref{eq:PH_theovector}), respectively, the density power divergence (DPD) between them,  as a function of a single tuning parameter $\beta\geq 0$, is given by 
\begin{align}
d_{\beta }(\widehat{\boldsymbol{p}}_{is},\boldsymbol{\pi }_{is}(\boldsymbol{\eta},\boldsymbol{\alpha})) =&\left( \pi_{is1}^{\beta+1}(\boldsymbol{\eta},\boldsymbol{\alpha})+\pi_{is2}^{\beta+1}(\boldsymbol{\eta},\boldsymbol{\alpha})\right) -\frac{\beta +1}{\beta }\left( \widehat{p}_{is1}\pi_{is1}^{\beta}(\boldsymbol{\eta},\boldsymbol{\alpha})+\widehat{p}_{is2}\pi_{is2}^{\beta}(\boldsymbol{\eta},\boldsymbol{\alpha})\right)  \notag \\
& +\frac{1}{\beta }\left( \widehat{p}_{is1}^{\beta +1}+\widehat{p}_{is2}^{\beta +1}\right) ,\quad \text{if }\beta >0,  \label{eq:PH_DPD_long}
\end{align}%
and $d_{\beta=0 }(\widehat{\boldsymbol{p}}_{is},\boldsymbol{\pi }_{is}(\boldsymbol{\eta},\boldsymbol{\alpha}))=\lim_{\beta \rightarrow 0^{+}}d_{\beta }(\widehat{\boldsymbol{p}}_{is},\boldsymbol{\pi }_{is}(\boldsymbol{\eta},\boldsymbol{\alpha}))=d_{KL}(\widehat{\boldsymbol{p}}_{is},\boldsymbol{\pi }_{is}(\boldsymbol{\eta},\boldsymbol{\alpha}))$, for $\beta =0$.

As the term $\frac{1}{\beta }\left( \widehat{p}_{is1}^{\beta +1}+\widehat{p}_{is2}^{\beta +1}\right) $ in (\ref{eq:PH_DPD_long}) has no role in the minimization with respect to $\boldsymbol{\theta}$,  we can consider the equivalent measure 
\begin{align*}
d_{\beta }^{\ast}(\widehat{\boldsymbol{p}}_{is},\boldsymbol{\pi }_{is}(\boldsymbol{\eta},\boldsymbol{\alpha})) =\left( \pi_{is1}^{\beta+1}(\boldsymbol{\eta},\boldsymbol{\alpha})+\pi_{is2}^{\beta+1}(\boldsymbol{\eta},\boldsymbol{\alpha})\right) -\frac{\beta +1}{\beta }\left( \widehat{p}_{is1}\pi_{is1}^{\beta}(\boldsymbol{\eta},\boldsymbol{\alpha})+\widehat{p}_{is2}\pi_{is2}^{\beta}(\boldsymbol{\eta},\boldsymbol{\alpha})\right),    \label{eq:PH_DPD_short1}
\end{align*}
and then can redefine the weighted minimum DPD estimator as follows.

\begin{definition}
The weighted minimum DPD estimator for $\boldsymbol{\theta}$ is given by
\begin{equation*}
\widehat{\boldsymbol{\theta}}_{\beta }=\underset{\boldsymbol{\theta}}{\arg \min }\sum_{i=1}^{I}\sum_{s=1}^{S}\frac{K_{is}}{K}d_{\beta }^{\ast}(\widehat{\boldsymbol{p}}_{is},\boldsymbol{\pi }_{is}(\boldsymbol{\eta},\boldsymbol{\alpha})),\quad \text{for }\beta >0,
\end{equation*}%
and for $\beta =0$, we have the MLE, $\widehat{\boldsymbol{\theta}}$, as defined in (\ref{eq:PH_mle_kull}).
\end{definition}

\section{Estimation and asymptotic distribution \label{sec:PH_esti} }
The estimating equations for the weighted minimum DPD estimator are as given in the following theorem.
\begin{theorem} \label{th:esti_eq}
For $\beta\geq0$, the estimating equations are given by%
\begin{align*}
\sum_{i=1}^{I}\sum_{s=1}^{S}&\delta_{is}(\boldsymbol{\eta}) \left(  K_{is}(1-R(IT_i,\boldsymbol{x}_s;\boldsymbol{\eta},\boldsymbol{\alpha})) -n_{is}\right) \\
&\times\left[(1-R(IT_i,\boldsymbol{x}_s;\boldsymbol{\eta},\boldsymbol{\alpha}))^{\beta-1}  + R^{\beta-1}(IT_i,\boldsymbol{x}_s;\boldsymbol{\eta},\boldsymbol{\alpha})  \right]  =\boldsymbol{0}_{I},\\
\sum_{i=1}^{I}\sum_{s=1}^{S}& \delta_{is}(\boldsymbol{\alpha})\left(  K_{is}(1-R(IT_i,\boldsymbol{x}_s;\boldsymbol{\eta},\boldsymbol{\alpha})) -n_{is}\right)\\
&\times \left[
(1-R(IT_i,\boldsymbol{x}_s;\boldsymbol{\eta},\boldsymbol{\alpha}))^{\beta-1}  + R^{\beta-1}(IT_i,\boldsymbol{x}_s;\boldsymbol{\eta},\boldsymbol{\alpha})  \right] =\boldsymbol{0}_{J},
\end{align*}
where
\begin{align}
\delta_{is}(\boldsymbol{\eta})=\frac{\partial R(IT_i,\boldsymbol{x}_s;\boldsymbol{\eta},\boldsymbol{\alpha})}{\partial \boldsymbol{\eta}}&=-(1-G_i)^{\lambda(\boldsymbol{x}_s;\boldsymbol{\alpha})-1}\lambda(\boldsymbol{x}_s;\boldsymbol{\alpha}) \frac{\partial G_i}{\partial \boldsymbol{\eta}}, \label{eq:PH_delta_eta}\\
\delta_{is}(\boldsymbol{\alpha})=\frac{\partial R(IT_i,\boldsymbol{x}_s;\boldsymbol{\eta},\boldsymbol{\alpha})}{\partial \boldsymbol{\alpha}}&=(1-G_i)^{\lambda(\boldsymbol{x}_s;\boldsymbol{\alpha})}log(1-G_i)\lambda(\boldsymbol{x}_s;\boldsymbol{\alpha})\boldsymbol{x}_s,\label{eq:PH_delta_alpha}
\end{align}
with 
\begin{equation}
\frac{\partial G_i}{\partial \eta_u}=\left\{\begin{array}{cc}
\exp(\eta_u)\exp(-\exp(\eta_u))G_i/\gamma(\eta_u) &, i\leq u,\\
0 &, i> u.
\end{array}\right. 
\end{equation}
\end{theorem}

\begin{proof}
The estimating equations are given by
\begin{align*}
\frac{\partial }{\partial \boldsymbol{\eta}}\sum_{i=1}^{I}\sum_{s=1}^{S}\frac{K_{is}}{K}d_{\beta }^{\ast }(\widehat{\boldsymbol{p}}_{is},\boldsymbol{\pi }_{is}(\boldsymbol{\eta},\boldsymbol{\alpha}))&=\sum_{i=1}^{I}\sum_{s=1}^{S}\frac{K_{is}}{K}\frac{\partial }{\partial \boldsymbol{\eta}}d_{\beta }^{\ast }(\widehat{\boldsymbol{p}}_{is},\boldsymbol{\pi }_{is}(\boldsymbol{\eta},\boldsymbol{\alpha}))=\boldsymbol{0}_{I},\\
\frac{\partial }{\partial \boldsymbol{\alpha}}\sum_{i=1}^{I}\sum_{s=1}^{S}\frac{K_{is}}{K}d_{\beta }^{\ast }(\widehat{\boldsymbol{p}}_{is},\boldsymbol{\pi }_{is}(\boldsymbol{\eta},\boldsymbol{\alpha}))&=\sum_{i=1}^{I}\sum_{s=1}^{S}\frac{K_{is}}{K}\frac{\partial }{\partial \boldsymbol{\alpha}}d_{\beta }^{\ast }(\widehat{\boldsymbol{p}}_{is},\boldsymbol{\pi }_{is}(\boldsymbol{\eta},\boldsymbol{\alpha}))=\boldsymbol{0}_{J},
\end{align*}
with
\begin{align}\label{eq:PH_estimatingProof}
&\frac{\partial }{\partial \boldsymbol{\eta}}d_{\beta }^{\ast }(\widehat{\boldsymbol{p}}_{is},\boldsymbol{\pi }_{is}(\boldsymbol{\eta},\boldsymbol{\alpha})) \notag \\
& =\left( \frac{\partial }{\partial \boldsymbol{\eta}}\pi _{is1}^{\beta +1}(\boldsymbol{\eta},\boldsymbol{\alpha})+\frac{\partial }{\partial \boldsymbol{\eta}}\pi _{is2}^{\beta +1}(\boldsymbol{\eta},\boldsymbol{\alpha})\right) -\frac{\beta +1}{\beta }\left( \widehat{p}_{is1}\frac{\partial }{\partial \boldsymbol{\eta}}\pi _{i1}^{\beta }(\boldsymbol{\theta})+\widehat{p}_{is2}\frac{\partial }{\partial \boldsymbol{\eta}}\pi _{is2}^{\beta }(\boldsymbol{\eta},\boldsymbol{\alpha})\right) \notag\\
& =\left( \beta +1\right) \left( \pi _{is1}^{\beta }(\boldsymbol{\eta},\boldsymbol{\alpha})-\pi_{is2}^{\beta }(\boldsymbol{\eta},\boldsymbol{\alpha})-\widehat{p}_{is1}\pi _{is1}^{\beta -1}(\boldsymbol{\eta},\boldsymbol{\alpha})+\widehat{p}_{is2}\pi _{is2}^{\beta -1}(\boldsymbol{\eta},\boldsymbol{\alpha})\right) \frac{\partial }{\partial \boldsymbol{\eta}}\pi _{is1}(\boldsymbol{\eta},\boldsymbol{\alpha}) \notag\\
& =\left( \beta +1\right) \left( \left( \pi _{i1}(\boldsymbol{\eta},\boldsymbol{\alpha})-\widehat{p}_{i1}\right) \pi _{is1}^{\beta -1}(\boldsymbol{\eta},\boldsymbol{\alpha})-\left( \pi _{is2}(\boldsymbol{\eta},\boldsymbol{\alpha})-\widehat{p}_{is2}\right) \pi _{is2}^{\beta -1}(\boldsymbol{\eta},\boldsymbol{\alpha})\right) \frac{\partial }{\partial \boldsymbol{\eta}}\pi _{is1}(\boldsymbol{\eta},\boldsymbol{\alpha})
\notag\\
& =\left( \beta +1\right) \left( \left( \pi _{is1}(\boldsymbol{\eta},\boldsymbol{\alpha})-\widehat{p}_{i1}\right) \pi _{is1}^{\beta -1}(\boldsymbol{\eta},\boldsymbol{\alpha})+\left( \pi _{i1}(\boldsymbol{\eta},\boldsymbol{\alpha})-\widehat{p}_{i1}\right) \pi _{is2}^{\beta -1}(\boldsymbol{\eta},\boldsymbol{\alpha})\right) \frac{\partial }{\partial \boldsymbol{\eta}}\pi _{is1}(\boldsymbol{\eta},\boldsymbol{\alpha}) \notag
\\
& =\left( \beta +1\right) \left( \pi _{is1}(\boldsymbol{\eta},\boldsymbol{\alpha})-\widehat{p}_{is1}\right) \left( \pi _{is1}^{\beta -1}(\boldsymbol{\eta},\boldsymbol{\alpha})+\pi _{is2}^{\beta-1}(\boldsymbol{\eta},\boldsymbol{\alpha})\right) \frac{\partial }{\partial \boldsymbol{\eta}}\pi_{is1}(\boldsymbol{\eta},\boldsymbol{\alpha})  
\end{align}%
and
\begin{align}
&\frac{\partial }{\partial \boldsymbol{\alpha}}d_{\beta }^{\ast }(\widehat{\boldsymbol{p}}_{is},\boldsymbol{\pi }_{is}(\boldsymbol{\eta},\boldsymbol{\alpha})) \notag \\
& =\left( \beta +1\right) \left( \pi _{is1}(\boldsymbol{\eta},\boldsymbol{\alpha})-\widehat{p}_{is1}\right) \left( \pi _{i1}^{\beta -1}(\boldsymbol{\eta},\boldsymbol{\alpha})+\pi _{is2}^{\beta-1}(\boldsymbol{\eta},\boldsymbol{\alpha})\right) \frac{\partial }{\partial \boldsymbol{\alpha}}\pi_{is1}(\boldsymbol{\eta},\boldsymbol{\alpha}).  
\end{align}
But, $\frac{\partial }{\partial \boldsymbol{\eta}}\pi_{is1}(\boldsymbol{\eta},\boldsymbol{\alpha})$ and $\frac{\partial }{\partial \boldsymbol{\alpha}}\pi_{is1}(\boldsymbol{\eta},\boldsymbol{\alpha})$ are as given in (\ref{eq:PH_delta_eta}) and (\ref{eq:PH_delta_alpha}), respectively. See equations (25) and (26) of \cite{ling2015_IEEERELIABILITY} for details.
\end{proof}

\vspace{0.3cm}

\begin{theorem} \label{th:asymp}
Let $\boldsymbol{\theta}^*$ be the true value of the parameter $\boldsymbol{\theta}$. Then, the asymptotic distribution of the weighted minimum DPD estimator, $\widehat{\boldsymbol{\theta}}_{\beta}$, is given by
\[
\sqrt{K}(  \widehat{\boldsymbol{\theta}}_{\beta}-\boldsymbol{\theta}^*)  \overset{\mathcal{L}}{\underset{K\mathcal{\rightarrow}\infty}{\longrightarrow}}\mathcal{N}\left(  \boldsymbol{0}_{I+J},\boldsymbol{{J}}_{\beta}^{-1}(\boldsymbol{\theta}^*)\boldsymbol{K}_{\beta}(\boldsymbol{\theta}^*)\boldsymbol{{J}}_{\beta}^{-1}(\boldsymbol{\theta}^*)\right)  ,
\]
where $\boldsymbol{{J}}_{\beta}(\boldsymbol{\theta})$ and $\boldsymbol{{K}}_{\beta}(\boldsymbol{\theta})$ are given by

\begin{align}
\boldsymbol{{J}}_{\beta}(\boldsymbol{\theta})&=\sum_{i=1}^{I}\sum_{s=1}^{S} \frac{K_{is}}{K}\Delta_{is}(\boldsymbol{\eta},\boldsymbol{\alpha})\left[(1-R(IT_i,\boldsymbol{x}_s;\boldsymbol{\eta},\boldsymbol{\alpha}))^{\beta-1}  + R^{\beta-1}(IT_i,\boldsymbol{x}_s;\boldsymbol{\eta},\boldsymbol{\alpha})  \right] \label{eq:PH_J} \\
\boldsymbol{{K}}_{\beta}(\boldsymbol{\theta})&=\sum_{i=1}^{I}\sum_{s=1}^{S} \frac{K_{is}}{K}\Delta_{is}(\boldsymbol{\eta},\boldsymbol{\alpha})(1-R(IT_i,\boldsymbol{x}_s;\boldsymbol{\eta},\boldsymbol{\alpha}))R(IT_i,\boldsymbol{x}_s;\boldsymbol{\eta},\boldsymbol{\alpha}) \notag \\
& \qquad  \qquad  \times \left[(1-R(IT_i,\boldsymbol{x}_s;\boldsymbol{\eta},\boldsymbol{\alpha}))^{\beta-1}  + R^{\beta-1}(IT_i,\boldsymbol{x}_s;\boldsymbol{\eta},\boldsymbol{\alpha})  \right]^2, \label{eq:PH_K}
\end{align}
with
\begin{align*}
\Delta_{is}(\boldsymbol{\eta},\boldsymbol{\alpha})=\left(\begin{array}{cc}
\delta_{is}(\boldsymbol{\eta})\delta_{is}^T(\boldsymbol{\eta})&\delta_{is}(\boldsymbol{\eta})\delta_{is}^T(\boldsymbol{\alpha}) \\
\delta_{is}(\boldsymbol{\alpha})\delta_{is}^T(\boldsymbol{\eta}) & \delta_{is}(\boldsymbol{\alpha})\delta_{is}^T(\boldsymbol{\alpha})
\end{array}\right),
\end{align*}
where $\delta_{is}(\boldsymbol{\eta})$ and $\delta_{is}(\boldsymbol{\alpha})$ are as given in (\ref{eq:PH_delta_eta}) and (\ref{eq:PH_delta_alpha}), respectively.
\end{theorem}

\begin{proof}
We denote
\begin{align*}
\boldsymbol{u}_{isj}(\boldsymbol{\eta},\boldsymbol{\alpha})  &  =\left(  \frac{\partial\log\pi_{isj}(\boldsymbol{\eta},\boldsymbol{\alpha})}{\partial\boldsymbol{\eta}},\frac{\partial\log\pi_{isj}(\boldsymbol{\eta},\boldsymbol{\alpha})}{\partial\boldsymbol{\alpha}}\right)  ^{T}\\
&= \left(\frac{1}{\pi_{isj}(\boldsymbol{\eta},\boldsymbol{\alpha})}\frac{\partial\pi_{isj}(\boldsymbol{\eta},\boldsymbol{\alpha})}{\partial\boldsymbol{\eta}},\frac{1}{\pi_{isj}(\boldsymbol{\eta},\boldsymbol{\alpha})}\frac{\partial\pi_{isj}(\boldsymbol{\eta},\boldsymbol{\alpha})}{\partial\boldsymbol{\alpha}}\right)  ^{T} \\
&=\left(  \frac{(-1)^{j+1}}{\pi_{isj}(\boldsymbol{\eta},\boldsymbol{\alpha})}\delta_{is}(\boldsymbol{\eta}),\frac{(-1)^{j+1}}{\pi_{isj}(\boldsymbol{\eta},\boldsymbol{\alpha})}\delta_{is}(\boldsymbol{\alpha})\right)  ^{T},
\end{align*}
with $\delta_{is}(\boldsymbol{\eta})$ and $\delta_{is}(\boldsymbol{\alpha})$ as given in  (\ref{eq:PH_delta_eta}) and (\ref{eq:PH_delta_alpha}), respectively.

Now, upon using Result 3.1 of \cite{ghosh2013_EJS}, we have

\[
\sqrt{K}\left(  \widehat{\boldsymbol{\theta}}_{\beta}-\boldsymbol{\theta}^{*}\right)  \overset{\mathcal{L}}{\underset{K\mathcal{\rightarrow}\infty}{\longrightarrow}}\mathcal{N}\left(  \boldsymbol{0}_{I+J},\boldsymbol{J}_{\beta}^{-1}(\boldsymbol{\theta}^{*})\boldsymbol{K}_{\beta}(\boldsymbol{\theta}^{*})\boldsymbol{J}_{\beta}^{-1}(\boldsymbol{\theta}^{*})\right)  ,
\]
where

\begin{align*}
\boldsymbol{J}_{\beta}(\boldsymbol{\theta})  &  =\sum_{i=1}^{I}\sum_{s=1}^{S}\sum_{j=1}^{2}\frac{K_{is}}{K}\boldsymbol{u}_{isj}(\boldsymbol{\eta},\boldsymbol{\alpha})\boldsymbol{u}_{isj}^{T}(\boldsymbol{\eta},\boldsymbol{\alpha})\pi_{isj}^{\beta+1}(\boldsymbol{\eta},\boldsymbol{\alpha}),\\
\boldsymbol{K}_{\beta}(\boldsymbol{\theta})  &  = \sum_{i=1}^{I}\sum_{s=1}^{S}\sum_{j=1}^{2}\frac{K_{is}}{K}\boldsymbol{u}_{isj}(\boldsymbol{\eta},\boldsymbol{\alpha})\boldsymbol{u}_{isj}^{T}(\boldsymbol{\eta},\boldsymbol{\alpha})\pi_{isj}^{2\beta+1}(\boldsymbol{\eta},\boldsymbol{\alpha}) -\sum_{i=1}^{I}\sum_{s=1}^{S}\frac{K_{is}}{K}\boldsymbol{\xi}_{is,\beta}(\boldsymbol{\eta},\boldsymbol{\alpha})\boldsymbol{\xi}_{is,\beta}^{T}(\boldsymbol{\eta},\boldsymbol{\alpha}),
\end{align*}
with
\begin{align*}
\boldsymbol{\xi}_{i,\beta}(\boldsymbol{\eta},\boldsymbol{\alpha})  &  =\sum_{j=1}^{2}\boldsymbol{u}_{isj}(\boldsymbol{\eta},\boldsymbol{\alpha})\pi_{isj}^{\beta+1}(\boldsymbol{\eta},\boldsymbol{\alpha}) =\left(  \delta_{is}(\boldsymbol{\eta}),\delta_{is}(\boldsymbol{\alpha})\right)  ^{T}\sum_{j=1}^{2}(-1)^{j+1}\pi_{isj}^{\beta}(\boldsymbol{\eta},\boldsymbol{\alpha}).
\end{align*}
Now, for $\boldsymbol{u}_{isj}(\boldsymbol{\eta},\boldsymbol{\alpha})\boldsymbol{u}_{isj}^{T}(\boldsymbol{\eta},\boldsymbol{\alpha})$, we have
\begin{align*}
\boldsymbol{u}_{isj}(\boldsymbol{\eta},\boldsymbol{\alpha})\boldsymbol{u}_{isj}^{T}(\boldsymbol{\eta},\boldsymbol{\alpha})&=\frac{1}{\pi_{isj}^{2}(\boldsymbol{\eta},\boldsymbol{\alpha})}
\left(\begin{array}{cc}
\delta_{is}(\boldsymbol{\eta})\delta_{is}^T(\boldsymbol{\eta})&\delta_{is}(\boldsymbol{\eta})\delta_{is}^T(\boldsymbol{\alpha}) \\
\delta_{is}(\boldsymbol{\alpha})\delta_{is}^T(\boldsymbol{\eta}) & \delta_{is}(\boldsymbol{\alpha})\delta_{is}^T(\boldsymbol{\alpha})
\end{array}\right) =\frac{1}{\pi_{ij}^{2}(\boldsymbol{\theta})}\Delta_{is}(\boldsymbol{\eta},\boldsymbol{\alpha}),
\end{align*}
with
\begin{align*}
\Delta_{is}(\boldsymbol{\eta},\boldsymbol{\alpha})=\left(\begin{array}{cc}
\delta_{is}(\boldsymbol{\eta})\delta_{is}^T(\boldsymbol{\eta})&\delta_{is}(\boldsymbol{\eta})\delta_{is}^T(\boldsymbol{\alpha}) \\
\delta_{is}(\boldsymbol{\alpha})\delta_{is}^T(\boldsymbol{\eta}) & \delta_{is}(\boldsymbol{\alpha})\delta_{is}^T(\boldsymbol{\alpha})
\end{array}\right).
\end{align*}
It then follows that%
\begin{align*}
\boldsymbol{J}_{\beta}(\boldsymbol{\theta})  &  =\sum_{i=1}^{I}\sum_{s=1}^{S}\frac{K_{is}}{K}\Delta_{is}(\boldsymbol{\eta},\boldsymbol{\alpha})\sum_{j=1}^{2}\pi_{isj}^{\beta-1}(\boldsymbol{\eta},\boldsymbol{\alpha}) \\
&=\sum_{i=1}^{I}\sum_{s=1}^{S}\frac{K_{is}}{K}\Delta_{is}(\boldsymbol{\eta},\boldsymbol{\alpha})\left(  \pi_{is1}^{\beta-1}(\boldsymbol{\eta},\boldsymbol{\alpha})+\pi_{is2}^{\beta-1}(\boldsymbol{\eta},\boldsymbol{\alpha})\right).
\end{align*}
\end{proof}

From here on, and for simplicity, we will  denote ${R}(IT_i, \boldsymbol{x}_0;{\boldsymbol{\eta}},\boldsymbol{\alpha})$ simply by ${R}(IT_i, \boldsymbol{x}_0;{\boldsymbol{\theta}}))$. Based on Theorem \ref{th:asymp}, the asymptotic variance of the weighted minimum DPD estimator of the reliability at inspection time $IT_i$ under normal operating condition $\boldsymbol{x}_0$ is given by 
$$
Var({R}(IT_i, \boldsymbol{x}_0;\widehat{\boldsymbol{\theta}}_{\beta}))\equiv Var({R}(\widehat{\boldsymbol{\theta}}_{\beta}))=\boldsymbol{P}^T \boldsymbol{\Sigma }_{\beta }( \widehat{\boldsymbol{\theta}}_{\beta })\boldsymbol{P},
$$
where
\begin{equation}\label{eq:PH_sigma}
\boldsymbol{\Sigma }_{\beta }( \widehat{\boldsymbol{\theta}}_{\beta })={\boldsymbol{J}}_{\beta }^{-1}( \widehat{\boldsymbol{\theta}}_{\beta }) {\boldsymbol{K}}_{\beta }( \widehat{\boldsymbol{\theta}}_{\beta }) {\boldsymbol{J}}_{\beta}^{-1}( \widehat{\boldsymbol{\theta}}_{\beta }) ,
\end{equation}
 ${\boldsymbol{J}}_{\beta }\left( \boldsymbol{\theta}\right) $, ${\boldsymbol{K}}_{\beta }\left( \boldsymbol{\theta}\right) $ are as given in (\ref{eq:PH_J}) and (\ref{eq:PH_K}), respectively, and $\boldsymbol{P}$ is a vector of the first-order derivates of ${R}(IT_i, \boldsymbol{x}_0;{\boldsymbol{\theta}}))$ with respect to the model parameters (see  (\ref{eq:PH_delta_eta}) and (\ref{eq:PH_delta_alpha})).  Consequently, the $100(1-\alpha)\%$ asymptotic confidence interval for the reliability function  $R(\boldsymbol{\theta})$ is given by
$$
\left(R(\widehat{\boldsymbol{\theta}}_{\beta})-z_{1-\alpha/2}se({R}(\widehat{\boldsymbol{\theta}}_{\beta})), \ R(\widehat{\boldsymbol{\theta}}_{\beta})+z_{1-\alpha/2}se({R}(\widehat{\boldsymbol{\theta}}_{\beta})) \right),
$$
where $se (R(\widehat{\boldsymbol{\theta}}_{\beta}))=\sqrt{Var (R(\widehat{\boldsymbol{\theta}}_{\beta}))}$ and $z_{\gamma}$ is the uppper $\gamma$ percentage point of the standard normal distribution.

However, an asymptotic confidence interval  may be satisfactory only for large sample sizes as it is based on the asymptotic properties of the estimators. 
\cite{balakrishnan2013_IEEEREL} found that, in the case of small sample sizes, the distribution of the MLE of the reliability is quite skewed, and so proposed a logit-transformation  for obtaining a  confidence interval for the reliability function, which can be extended to the case of  the weighted minimum DPD estimators of the reliabilities as well to obtain a confidence interval of the form:
\begin{equation}
\left(\frac{R(\widehat{\boldsymbol{\theta}}_{\beta})}{R(\widehat{\boldsymbol{\theta}}_{\beta})+(1-R(\widehat{\boldsymbol{\theta}}_{\beta}))T}, \ \frac{R(\widehat{\boldsymbol{\theta}}_{\beta})}{R(\widehat{\boldsymbol{\theta}}_{\beta})+(1-R(\widehat{\boldsymbol{\theta}}_{\beta}))/T} \right),
\end{equation}
where $T=\exp\left(z_{1-\alpha/2}\frac{se({R}(\widehat{\boldsymbol{\theta}}_{\beta}))}{R(\widehat{\boldsymbol{\theta}}_{\beta})(1-R(\widehat{\boldsymbol{\theta}}_{\beta}))}\right)$.
 
\section{Wald-type tests \label{sec:PH_wald}}
Let us consider the function $\boldsymbol{m}:\mathbb{R}^{I+J}\longrightarrow \mathbb{R}^{r}$, where $r\leq (I+J)$ and
\begin{equation}\label{eq:testA}
\boldsymbol{m}\left( \boldsymbol{\theta}\right) =\boldsymbol{0}_{r},
\end{equation}
which corresponds to a composite null hypothesis. We assume that the $(I+J) \times r$ matrix \ $\boldsymbol{M}( \boldsymbol{\theta}) =\frac{\partial \boldsymbol{m}^{T}\left( \boldsymbol{\theta}\right) }{\partial \boldsymbol{\theta}}$
exists and is continuous in  $\boldsymbol{\theta}$ and rank $\boldsymbol{M}\left( \boldsymbol{\theta}\right) =r$.  Then, for testing
\begin{equation}
H_{0}:\boldsymbol{\theta\in \Theta }_{0}\text{ against }H_{1}:\boldsymbol{\theta\notin\Theta }_{0},  \label{eq:PH_HComp}
\end{equation}
where  $\boldsymbol{\Theta }_{0}=\left\{ \boldsymbol{\theta}\in \mathbb{R}^{(I+J)} :\boldsymbol{m}\left( \boldsymbol{\theta}\right) =\boldsymbol{0}_{r}\right\}$, we can consider the following Wald-type test statistics:
\begin{align}\label{eq:PH_waldeq}
W_{K}( \widehat{\boldsymbol{\theta}}_{\beta }) =K\boldsymbol{m}^{T}( \widehat{\boldsymbol{\theta}}_{\beta }) \left( \boldsymbol{M}^{T}( \widehat{\boldsymbol{\theta}}_{\beta }) \boldsymbol{\Sigma }( \widehat{\boldsymbol{\theta}}_{\beta }) \boldsymbol{M}( \widehat{\boldsymbol{\theta}}_{\beta }) \right) ^{-1}\boldsymbol{m}( \widehat{\boldsymbol{\theta}}_{\beta }) ,
\end{align}%
where\ $\boldsymbol{\Sigma }_{\beta }( \widehat{\boldsymbol{\theta}}_{\beta })$
is  as given in (\ref{eq:PH_sigma}).


\begin{theorem}
\label{th:test_asym} Under (\ref{eq:testA}), we have 
\begin{equation*}
W_{K}( \widehat{\boldsymbol{\theta}}_{\beta }) \underset{K\rightarrow
\infty }{\overset{\mathcal{L}}{\longrightarrow }}\chi _{r}^{2},
\end{equation*}
where $\chi_{r}^2$ denotes a central chi-square distribution with $r$ degrees of freedom.
\end{theorem}

\begin{proof}
Let $\boldsymbol{\theta}^{0}\in \Theta_0 $ be the true value of the parameter $\boldsymbol{\theta}$. It is clear that 
\begin{align*}
\boldsymbol{m}\left( \widehat{\boldsymbol{\theta}}_{\beta }\right) &=\boldsymbol{m}\left( \boldsymbol{\theta}^{0}\right)  + \boldsymbol{M}^{T}\left( \widehat{\boldsymbol{\theta}}_{\beta }\right) \left( \widehat{\boldsymbol{\theta}}_{\beta }-\boldsymbol{\theta}^{0}\right) +o_{p}\left( \left\Vert \widehat{\boldsymbol{\theta}}_{\beta }-\boldsymbol{\theta}^{0}\right\Vert \right) \\
&=\boldsymbol{M}^{T}\left( \widehat{\boldsymbol{\theta}}_{\beta }\right) \left( \widehat{\boldsymbol{\theta}}_{\beta }-\boldsymbol{\theta}^{0}\right) +o_{p}\left(K^{-1/2}\right) .
\end{align*}
But, under $H_{0}$,\ $\sqrt{K}\left( \widehat{\boldsymbol{\theta}}_{\beta }-\boldsymbol{\theta}^{0}\right) \underset{K\rightarrow \infty }{\overset{\mathcal{L}}{\longrightarrow }}\mathcal{N}\left( \boldsymbol{0}_{(I+J)},\boldsymbol{\Sigma }_{\beta }\left( \boldsymbol{\theta}^0\right) \right)$.
Therefore, under $H_{0}$,
\[
\sqrt{K}\boldsymbol{m}\left( \widehat{\boldsymbol{\theta}}_{\beta }\right) \underset{K\rightarrow \infty }{\overset{\mathcal{L}}{\longrightarrow }}\mathcal{N}\left( \boldsymbol{0}_{r},\boldsymbol{M}^{T}\left( \boldsymbol{\theta}^{0}\right) \boldsymbol{\Sigma }_{\beta }\left( \boldsymbol{\theta}^0\right) \boldsymbol{M}\left( \boldsymbol{\theta}^{0}\right) \right)
\]
and taking into account that $rank(\boldsymbol{M}\left( \boldsymbol{\theta}^{0}\right) )=r$, we obtain
\begin{equation*}
K\boldsymbol{m}^{T}\left( \widehat{\boldsymbol{\theta}}_{\beta }\right) \left( \boldsymbol{M}^{T}\left( \boldsymbol{\theta}^{0}\right) \boldsymbol{\Sigma }_{\beta }\left( \boldsymbol{\theta}^{0}\right) \boldsymbol{M}\left( \boldsymbol{\theta}^{0}\right) \right) ^{-1}\boldsymbol{m}\left( \widehat{\boldsymbol{\theta}}_{\beta }\right) \underset{K\rightarrow \infty }{\overset{\mathcal{L}}{\longrightarrow }}\chi _{r}^{2}.
\end{equation*}%
Because $\left( \boldsymbol{M}^{T}\left( \widehat{\boldsymbol{\theta}}_{\beta }\right)
\boldsymbol{\Sigma }_{\beta }\left( \widehat{\boldsymbol{\theta}}_{\beta
}\right) \boldsymbol{M}\left( \widehat{\boldsymbol{\theta}}_{\beta }\right)
\right) ^{-1}$ is a consistent estimator of $\left( \boldsymbol{M}^{T}\left( \boldsymbol{\theta}^{0}\right) \boldsymbol{\Sigma }_{\beta }\left( \boldsymbol{\theta}^{0}\right) \boldsymbol{M}\left( \boldsymbol{\theta}^{0}\right)\right) ^{-1}$, we get 
\begin{equation*}
W_{K}\left( \widehat{\boldsymbol{\theta}}_{\beta }\right) \underset{K\rightarrow\infty }{\overset{\mathcal{L}}{\longrightarrow }}\chi _{r}^{2}.
\end{equation*}
\end{proof}

Based on Theorem \ref{th:test_asym}, we shall reject the null hypothesis  in (\ref{eq:PH_HComp}) if
\begin{equation}
W_{K}( \widehat{\boldsymbol{\theta}}_{\beta }) >\chi _{r,\alpha }^{2},
\label{eq:PH_reject}
\end{equation}%
where $\chi^2_{r,\alpha}$ is the upper $\alpha$ percentage point of $\chi_r^2$ distribution.

Wald-type test statistics based on weighted minimum DPD estimators have been considered previously by a number of authors including  \cite{basu2016_STATISTICS, castilla2018_BIOMETRICS} and \cite{balakrishnan2019_IEEEIT, balakrishnan2019_METRIKA, balakrishnan2019_IEEERELIABILITY}.

\section{Monte Carlo Simulation Results \label{sec:PH_sim}}

In this section, an extensive simulation study is carried out for evaluating  the proposed weighted minimum DPD estimators and Wald-type tests. The simulations results are computed based on $1,000$  simulated samples in the R statistical software.  Mean square error (MSE) and bias are computed for evaluating the estimators in both balanced and  unbalanced data sets, while empirical levels and powers are computed for evaluating the tests.

\subsection{Weighted minimum DPD estimators}
Suppose  the lifetimes of test units follow a Weibull distribution (see Remark \ref{rem:PH_Weibull}). All the test units were divided into $S=4$ groups, subject to different acceleration conditions with $J=2$ stress factors  at two elevated stress levels each, that is, $(x_1 ,x_2 ) =\{(55,70),(55,100),(85,70), (85,100)\}$, and were inspected at $I=3$ different times, \newline $(IT_1 ,IT_2 ,IT_3) = (2,5,8)$.

\subsubsection{Balanced data}
We assume $(c_1,c_2)=(-0.03,-0.03)$, $c_0 \in \{6, 6.5 \}$ for different degrees of reliability and $b \in \{0,0.5 \}$. Note that the exponential distribution will be included as a special case when we take $b=0$.  In this framework, we consider “outlying cells” rather than “outlying observations”. A cell which does not follow the one-shot device model will be called an outlying cell or outlier.  In this cell, the number of devices failed will be different than what is expected. This is inthe spirit of principle of inflated models in distribution theory (see Lambert (1992) and Heilbron (1994)). This outlying cell (taken to be $i=3$, $s=4$), is generated under the parameters  $(\tilde{c}_1,\tilde{c}_2)=(-0.027,-0.027)$ and $\tilde{b}\in\{0.05,0.45\}$.

Bias of estimates are then computed for different (equal) samples sizes $K_{is} \in \{50,70,100\}$ and tuning parameters $\beta \in \{0,0.2,0.4,0.6\}$ for both pure and contaminated data. The obtained results are presented in Tables \ref{table:PH_6000}, \ref{table:PH_6005}, \ref{table:PH_6500} and \ref{table:PH_6505}. As expected, when  the sample size increases, errors tend to decrease, while in the contaminated data set, these errors are generally greater than in the case of uncontaminated data. Weighted minimum DPD estimators with $\beta>0$ present a better behaviour than the MLE in terms of robustness.  Note that  reliabilities are underestimated and that the estimates are quite precise in all the cases.

\subsubsection{Unbalanced data \label{sec:PH_unbalanced}}
In this setting, we consider an unbalanced data set, in which at each inspection time $i$, $(K_{i1},K_{i2}, K_{i3}, K_{i4})=(10r,15r,20r,30r)$ for different values of the factor $r \in \{1,2, \dots,10 \}$. We then assume  $(c_0, c_1,c_2)=(6,5, -0.03,-0.03)$, $b=0.5$, and $\tilde{c}_2=-0.027$. MSEs of the parameter $\boldsymbol{\theta}$ are then computed and the obtained results are presented in  Figure \ref{fig:PH_unbalanced}.

As expected, when the sample size increases, the MSE decreases, but lack of robustness of the MLE ($\beta=0$) as compared to the weighted minimum DPD estimators with $\beta>0$ becomes quite evident. 

\subsection{Wald-type tests}

To evaluate the performance of the proposed Wald-type tests, we consider the scenario of unbalanced data proposed discussed above. We consider the testing problem
\begin{equation}
H_0 : \alpha_1 = 0.04946 \quad \text{against}  \quad H_1 : \alpha_1 \neq  0.04946,
\end{equation}

Under the same simulation scheme as used above in Section \ref{sec:PH_unbalanced}, we first evaluate the empirical levels, measured as the proportion of Wald-type test statistics exceeding the corresponding chi-square critical value for a nominal size of $0.05$. The empirical powers are computed in a similar manner, with $\alpha_1^0=0.05276$ ($c_1=-0.032$, $c_2=-0.028$). The obtained results are shown in Figure \ref{fig:PH_unbalanced_Wald}. 

In the case of uncontaminated data, the conventional Wald test has level to be close to nominal value and also has good power performance.  The robust tests, however, has a slightly inflated level values (as compared to the nominal value), but possesses similar power as the conventional Wald test (which is evident from the Figure \ref{fig:PH_unbalanced_Wald}).  But, when the data is contaminated, the level of the conventional Wald test turns out to be quite non-robust and takes on very high values as compared to the nominal level.  This, in turn, results in higher power (see Figure \ref{fig:PH_unbalanced_Wald}).  However, the proposed robust tests maintain levels close to the nominal value and also possesses good power values (as can be seen in the Figure \ref{fig:PH_unbalanced_Wald}).  Thus, taking both level and power into account, the robust tests, though is slightly inferior to the conventional Wald test in the case of uncontaminated data, turn out to be considerably more efficient than the conventional Walk test in the case of contaminated data

\section{Application to Real Data \label{sec:PH_data}}

\subsection{Testing on proportional Hazard rates}
Based on \cite{balakrishnan2012_IEEEREL}, we suggest a distance-based statistic on the form 
\begin{equation}\label{eq:PH_Mbeta}
M_{\beta}=max_{i,s}\left|n_{is}-K_{is}(1-R(IT_i,\boldsymbol{x}_s;\widehat{\boldsymbol{\theta}}_{\beta})) \right|
\end{equation}
as a discrepancy measure for evaluating the fit of the assumed model to the observed data. If the assumed model is not a good fit to the data, we will obtain a large value of $M_{\beta}$. In fact, under the assumed model, we have
$$
n_{is}\sim \text{Binomial} (K_{is},1-R(IT_i,\boldsymbol{x}_s;{\boldsymbol{\theta}}))´,
$$
and so, by denoting $\Phi_{is}=\lceil K_{is}(1-R(IT_i,\boldsymbol{x}_s;\widehat{\boldsymbol{\theta}}_{\beta}))- M_{\beta}\rceil$ and $\Psi_{is}=\lfloor K_{is}(1-R(IT_i,\boldsymbol{x}_s;\widehat{\boldsymbol{\theta}}_{\beta}))+ M_{\beta}\rfloor$, the corresponding exact p-value is given by
\begin{align}\label{eq:PH_pvalue}
p-value&=Pr\left(max_{i,s}\left|n_{is}-K_{is}(1-R(IT_i,\boldsymbol{x}_s;\widehat{\boldsymbol{\theta}}_{\beta})) \right| >M_{\beta}\right)\notag \\
&=1-Pr\left(max_{i,s}\left|n_{is}-K_{is}(1-R(IT_i,\boldsymbol{x}_s;\widehat{\boldsymbol{\theta}}_{\beta})) \right| \leq M_{\beta}\right) \notag \\
&=1-\prod_{i=1}^I\prod_{s=1}^S Pr\left(\left|n_{is}-K_{is}(1-R(IT_i,\boldsymbol{x}_s;\widehat{\boldsymbol{\theta}}_{\beta})) \right| \leq M_{\beta}\right) \notag \\
&=1-\prod_{i=1}^I\prod_{s=1}^SPr\left(\Phi_{is}\leq n_{is}\leq\Psi_{is} \right).
\end{align}
From (\ref{eq:PH_pvalue}), we can readily validate the proportional hazards assumption if the $p$-value is sufficiently large.

\subsection{Choice of the tuning parameter}

In the preceding discussion, we have seen how weighted minimum DPD estimators with $\beta>0$ tend to be more robust than the classical MLE overall whencontamination is present in the data. MLE has been shown to be  more efficient when there is no contamination in the data. It is then necessary to provide a data-driven procedure for the determination of the optimal choice of the tuning parameter that would provide a trade-off between efficiency and robustness. One way to do this is as follows: In a grid  of possible tuning parameters, apply a measure of discrepancy to the data. Then, the tuning parameter that  leads to the minimum discrepancy-statistic can be chosen as  the ``optimal'' one.

A possible choice of the discrepancy measure could be $M_{\beta}$,  given in (\ref{eq:PH_Mbeta}).  Another idea may be by minimizing the estimated mean square error. This method, originally proposed by \cite{warwick2005_JSCS}, was applied in the context of one-shot devices in \cite{balakrishnan2019_IEEEIT,balakrishnan2019_IEEERELIABILITY}. The estimation of the MSE is as follows:
\begin{align*}
\widehat{MSE}(\beta)&=(\boldsymbol{\theta}_{\beta}-\boldsymbol{\theta}_P)^T(\boldsymbol{\theta}_{\beta}-\boldsymbol{\theta}_P)+\frac{1}{K}\text{trace}\left\{\boldsymbol{J}^{-1} _{\beta}(\boldsymbol{\theta}_{\beta})\boldsymbol{K} _{\beta}(\boldsymbol{\theta}_{\beta})\boldsymbol{J}^{-1} _{\beta}(\boldsymbol{\theta}_{\beta})\right\},
\end{align*}
where $\boldsymbol{\theta}_P$ is a pilot estimator, whose choice will affect the overall procedure. If we take $\boldsymbol{\theta}_P=\widehat{\boldsymbol{\theta}}_{\beta}$, the approach  coincides with that of \cite{hong2001_JKSS}, but it does not take into account the model misspecification. Note that, the need for a pilot estimator becomes a drawback of this procedure, as will be seen in the next section.

\subsection{Electric Current data}

We now consider the Electric Current data  (\cite{ling2015_IEEERELIABILITY}), in which 120 one-shot devices  were divided into four accelerated conditions with higher-than-normal temperature and electric current, and inspected at three different times (see Table \ref{table:PH_examplebalaI}).

In Table \ref{table:PH_examplebalaIb}, estimates of the model parameters by the use of the proportional hazards model and the Weibull distribution (see \cite{balakrishnan2019_IEEERELIABILITY}) are provided, for different values of the tuning parameter. Estimates of reliabilities and confidence intervals under the proportional hazards assumption are given in Table \ref{table:PH_examplebalaIc}.

Table \ref{table:PH_examplebalaIb} also presents the dvalues of the distance-statistic $M_{\beta}$ and the corresponding $p$-values. From these values, it seems that the proportional hazards assumption fits the data at least as well as the Weibull model. The best fit is obtained for $\beta=0.5$. To complete the study,  \cite{warwick2005_JSCS} approach is achieved for different values of the pilot estimator in a grid of width  $100$. However, as pointed out before, the final choice of the optimal tuning parameter depends too much on the pilot estimator used (see Figure \ref{fig:PH_betaopt}).   Recently, \cite{basak2020} proposed an ``iterated Warwick and Jones algorithm'' trying to solve this problem.

\begin{table}[h!]\setlength{\tabcolsep}{2.5pt}  \renewcommand{\arraystretch}{1.5}
\caption{Electric Current data  \label{table:PH_examplebalaI}}
\center

\begin{tabular}{r cccccccccccc}
\hline 
Inspection Time $IT_i$ \ & 2 & 2 & 2 & 2 & 5 & 5 & 5 & 5 & 8 & 8 & 8 & 8 \\ 
Temperature $x_{s1}$ \ & 55 & 80 & 55 & 80 & 55 & 80 & 55 & 80 & 55 & 80 & 55 & 80 \\ 
Electric current $x_{s2}$ \ & 70 & 70 & 100 & 100 & 70 & 70 & 100 & 100 & 70 & 70 & 100 & 100 \\ 
Number of failures $n_{is}$ \ & 4 & 8 & 9 & 8 & 7 & 9 & 9 & 9 & 6 & 10 & 9 & 10 \\ 
Number of tested items $K_{is}$ \ & 10 & 10 & 10 & 10 & 10 & 10 & 10 & 10 & 10 & 10 & 10 & 10 \\ 
\hline 
\end{tabular} 
\end{table}

 \begin{table}[h!]\setlength{\tabcolsep}{2.2pt}\renewcommand{\arraystretch}{1.2}
 \caption{Electric Current data: one-shot device testing data analysis by using the proportional hazards model and the Weibull distribution \label{table:PH_examplebalaIb}}
 \center
 \small
\begin{tabular}{l rrrrrrrrrrrrrr}
\hline
 &  \multicolumn{8}{c}{ Proportional Hazards model}&  \multicolumn{6}{c}{Weibull distribution }\\ 
   \cline{2-8} \cline{10-15}
 $\beta$   &$M_\beta$& p-value& $T^\circ$ & current & $\eta_1$  & $\eta_2$  & $\eta_3$  &  &$M_\beta$  & p-value     &intercept & $T^\circ$    & current     & shape    \\    \cline{2-8} \cline{10-15}

0\  & 1.80 & 0.695&0.023 & 0.018 & 0.123 & 0.543 & -2.182 &  &1.80&0.695& 7.022  & -0.053 & -0.040 & -0.817 \\
0.1\ &1.72 &0.745&0.024 & 0.018 & 0.141 & 0.555 & -2.283 &  &1.72&0.745& 7.398  & -0.055 & -0.043 & -0.845 \\
0.2\ &1.65 &0.796&0.024 & 0.019 & 0.156 & 0.565 & -2.399 &  &1.65&0.796& 7.803  & -0.057 & -0.046 & -0.869 \\
0.3 \ &1.58 &0.833&0.025 & 0.020 & 0.167 & 0.572 & -2.534 &  &1.57&0.833& 8.254  & -0.060 & -0.050 & -0.890 \\
0.4\ &1.49 &0.931&0.026 & 0.022 & 0.177 & 0.579 & -2.695 &  &1.49&0.931& 8.747  & -0.064 & -0.054 & -0.906 \\
0.5\ &1.40 &0.942&0.027 & 0.023 & 0.183 & 0.582 & -2.887 &  &1.40&0.942& 9.324  & -0.068 & -0.058 & -0.920 \\
0.6\ &1.51 &0.892&0.029 & 0.025 & 0.187 & 0.585 & -3.130 &  &1.51&0.892& 10.026 & -0.073 & -0.063 & -0.931 \\
0.7\ &1.64 &0.876&0.031 & 0.027 & 0.190 & 0.586 & -3.438 &  &1.64&0.876& 10.868 & -0.079 & -0.069 & -0.938 \\
0.8\ &1.76 &0.861&0.033 & 0.030 & 0.189 & 0.586 & -3.798 &  &1.76&0.861& 11.827 & -0.086 & -0.076 & -0.942 \\
0.9\ &1.84 &0.750&0.036 & 0.032 & 0.185 & 0.582 & -4.106 &  &1.84&0.750& 12.575 & -0.091 & -0.082 & -0.938 \\ \hline
\end{tabular}
\end{table}

\begin{figure}[h!!!]
\centering
\begin{tabular}{cc}
\includegraphics[scale=0.42]{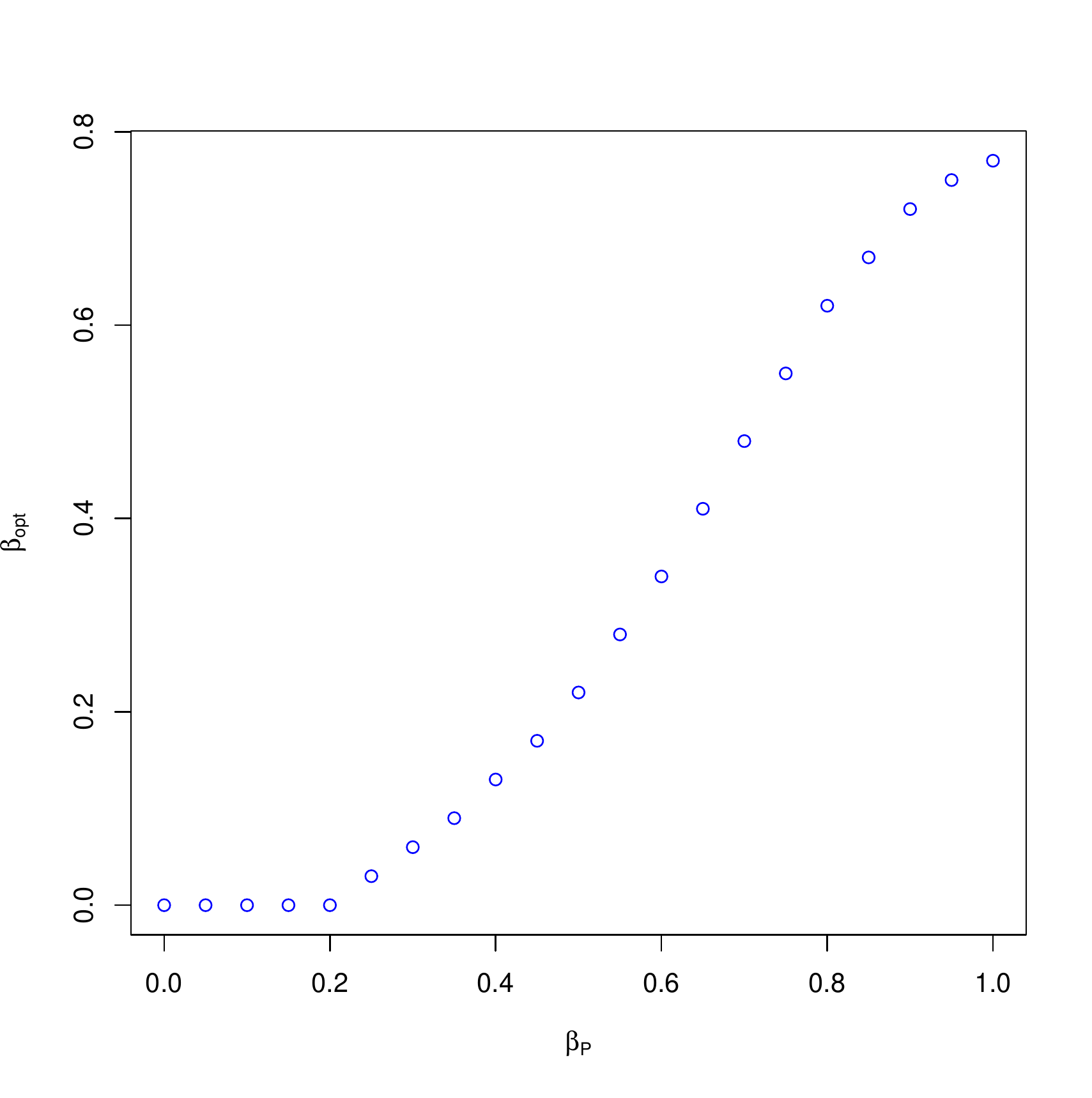} & 
\end{tabular}
\caption{ Electric Current data: estimation of the optimal tuning parameter depending on a pilot estimator by \cite{warwick2005_JSCS} procedure \label{fig:PH_betaopt}}
\end{figure}

\begin{table}[h!]\setlength{\tabcolsep}{3.7pt}\renewcommand{\arraystretch}{1.2}
 \caption{Electric Current data: estimates of reliabilities and corresponding confidence intervals \label{table:PH_examplebalaIc}}
 \center
\begin{tabular}{r llll}
\hline
$\beta$ &  \multicolumn{1}{c}{ $R(2,25,35;\widehat{\boldsymbol{\theta}}_{\beta})$}&  \multicolumn{1}{c}{$R(5,25,35;\widehat{\boldsymbol{\theta}}_{\beta})$} &  \multicolumn{1}{c}{$R(8,25,35;\widehat{\boldsymbol{\theta}}_{\beta})$}\\  \hline
0   & 0.817 (0.516, 0.949) & \ 0.739 (0.397, 0.924) & \ 0.689 (0.336, 0.907) \\
0.1 & 0.824 (0.526, 0.952) & \ 0.751 (0.412, 0.928) & \ 0.704 (0.353, 0.912)\\
0.2 & 0.833 (0.535, 0.956) & \ 0.765 (0.427, 0.934) & \ 0.721 (0.370, 0.919) \\
0.3 & 0.843 (0.545, 0.960) & \ 0.780 (0.442, 0.941) & \ 0.740 (0.387, 0.927) \\
0.4 & 0.855 (0.555, 0.965) & \ 0.797 (0.457, 0.948) & \ 0.760 (0.405, 0.936) \\
0.5 & 0.868 (0.566, 0.971) & \ 0.816 (0.474, 0.956) & \ 0.782 (0.423, 0.946) \\
0.6 & 0.884 (0.581, 0.976) & \ 0.837 (0.493, 0.965) & \ 0.807 (0.445, 0.956) \\
0.7 & 0.901 (0.601, 0.982) & \ 0.861 (0.516, 0.973) & \ 0.836 (0.471, 0.967) \\
0.8 & 0.918 (0.626, 0.987) & \ 0.885 (0.544, 0.980) & \ 0.863 (0.503, 0.975) \\
0.9 & 0.931 (0.649, 0.990) & \ 0.902 (0.570, 0.985) & \ 0.884 (0.533, 0.981) \\
 \hline
\end{tabular}
\end{table}

\section{Concluding Remarks \label{sec:PH_conc}}
In this paper, we have developed new estimators and tests for one-shot device testing under proportional hazards assumption. This semi-parametric model is presented as an alternative to  parametric models, by allowing the hazard rate increasing in a non-parametric way. An extensive simulation study carried out shows the robustness of the proposed methods of inference. Because model selection is an important part of reliability analysis, a test statistic for checking the proportional hazards assumption is presented as well and applied to the numerical example.

As future work, it would be of interest to develop model selection criteria, and also to extend the proposed method to the case of competing risks problem, when there is more than one cause of failure of one-shot devices.

\begin{table}[p]\setlength{\tabcolsep}{2.1pt}  \renewcommand{\arraystretch}{1.7}
\caption{Bias for the semi-parametric model with $b=0$ and $c_0=6$. \label{table:PH_6000}}
\centering
\small
\begin{tabular}{r r c rrrr c rrrr}
     &          &&          &          &          &           &&          &          &             &           \\ \hline
  & & \multicolumn{4}{c}{ Uncontaminated  data}& & \multicolumn{4}{c}{Contaminated data}\\ 
 \cline{4-7} \cline{9-12}

$K_{is}=50$             & True value    && 0        & 0.2      & 0.4      & 0.6       && 0        & 0.2      & 0.4         & 0.6       \\ \hline 
$\eta_1$        & -0.66688 && -0.00494 & -0.00276 & -0.00053 & -0.00372 && 0.09898  & 0.06722  & 0.03547  & 0.01708  \\
$\eta_2$         & -0.01304 && -0.00228 & -0.00078 & 0.00109  & -0.00131 && 0.06902  & 0.04716  & 0.02531  & 0.01286  \\
$\eta_3$       & -3.92056 && -0.02788 & -0.01810 & -0.01389 & -0.01982 && 0.34916  & 0.23252  & 0.12087  & 0.05402  \\
$\alpha_1$      & 0.03000  && 0.00010  & 0.00002  & -0.00001 & 0.00002  && -0.00281 & -0.00193 & -0.00107 & -0.00056 \\
$\alpha_2$     & 0.03000  && 0.00033  & 0.00027  & 0.00025  & 0.00030  && -0.00259 & -0.00167 & -0.00081 & -0.00028 \\
$R(15,\boldsymbol{x}_0)$& 0.79857  && -0.00520 & -0.00611 & -0.00686 & -0.00671 && -0.03387 & -0.02502 & -0.01652 & -0.01202 \\
\hline 
  & & \multicolumn{4}{c}{ Uncontaminated  data}& & \multicolumn{4}{c}{Contaminated data}\\ 
 \cline{4-7} \cline{9-12}
$K_{is}=70$               &True value    && 0        & 0.2      & 0.4      & 0.6       && 0        & 0.2      & 0.4         & 0.6       \\ \hline 
$\eta_1$        & -0.66688 && -0.00780 & -0.00675 & -0.00716 & -0.00802 && 0.09876  & 0.06233   & 0.03209  & 0.01278  \\
$\eta_2$        & -0.01304 && -0.00459 & -0.00386 & -0.00410 & -0.00465 && 0.06810  & 0.04315   & 0.02257  & 0.00948  \\
$\eta_3$        & -3.92056 && -0.03954 & -0.03763 & -0.04073 & -0.04458 && 0.35084  & 0.21624   & 0.10254  & 0.03023  \\
$\alpha_1$      & 0.03000  && 0.00027  & 0.00025  & 0.00026  & 0.00029  && -0.00276 & -0.00173  & -0.00086 & -0.00030 \\
$\alpha_2$     & 0.03000  && 0.00035  & 0.00034  & 0.00038  & 0.00041  && -0.00267 & -0.00163  & -0.00075 & -0.00017 \\
$R(15,\boldsymbol{x}_0)$ & 0.79857  && -0.00197 & -0.00221 & -0.00223 & -0.00221 && -0.03131 & -0.02076  & -0.01246 & -0.00749  \\ 
\hline 
  & & \multicolumn{4}{c}{ Uncontaminated  data}& & \multicolumn{4}{c}{Contaminated data}\\ 
 \cline{4-7} \cline{9-12}
$K_{is}=100$              & True value    && 0        & 0.2      & 0.4      & 0.6       && 0        & 0.2      & 0.4         & 0.6       \\ \hline 
$\eta_1$        & -0.66688 && -0.00778 & -0.00682 & -0.00711 & -0.00785 && 0.09857  & 0.06207   & 0.03231  & 0.01320  \\
$\eta_2$         & -0.01304 && -0.00477 & -0.00412 & -0.00429 & -0.00477 && 0.06776  & 0.04275   & 0.02248  & 0.00952  \\
$\eta_3$         & -3.92056 && -0.02739 & -0.02332 & -0.02387 & -0.02586 && 0.36315  & 0.23019   & 0.12013  & 0.04993  \\
$\alpha_1$      & 0.03000  && 0.00031  & 0.00028  & 0.00029  & 0.00031  && -0.00271 & -0.00169  & -0.00084 & -0.00029 \\
$\alpha_2$     & 0.03000  && 0.00016  & 0.00013  & 0.00013  & 0.00015  && -0.00287 & -0.00185  & -0.00099 & -0.00045 \\
$R(15,\boldsymbol{x}_0)$ & 0.79857  && -0.00144 & -0.00176 & -0.00185 & -0.00187 && -0.03085 & -0.02034  & -0.01214 & -0.00720\\  \hline 
\end{tabular}
\end{table}

\begin{table}[p]\setlength{\tabcolsep}{2.1pt}  \renewcommand{\arraystretch}{1.7}
\caption{Bias for the semi-parametric model with $b=0.5$ and $c_0=6$. \label{table:PH_6005} }
\centering
\small
\begin{tabular}{r r c rrrr c rrrr}
     &          &&          &          &          &           &&          &          &             &           \\ \hline
  & & \multicolumn{4}{c}{ Uncontaminated  data}& & \multicolumn{4}{c}{Contaminated data}\\ 
 \cline{4-7} \cline{9-12}

$K_{is}=50$             & True value    && 0        & 0.2      & 0.4      & 0.6       && 0        & 0.2      & 0.4         & 0.6       \\ \hline 
$\eta_1$         & -1.38827 && -0.01224 & -0.00988 & -0.03700 & -0.07648 && 0.03590  & -0.00540 & -0.03356 & -0.08611 \\
$\eta_2$        & -0.48138 && -0.00687 & -0.00537 & -0.02153 & -0.04394 && 0.02362  & -0.00239 & -0.01962 & -0.04962 \\
$\eta_3$        & -6.46391 && -0.05973 & -0.05148 & -0.19693 & -0.40632 && 0.14538  & -0.03531 & -0.17304 & -0.47276 \\
$\alpha_1$      & 0.04946  && 0.00032  & 0.00023  & 0.00124  & 0.00274  && -0.00125 & 0.00010  & 0.00106  & 0.00325  \\
$\alpha_2$      & 0.04946  && 0.00061  & 0.00056  & 0.00174  & 0.00361  && -0.00097 & 0.00043  & 0.00155  & 0.00410  \\
$R(15,\boldsymbol{x}_0)$& 0.91810  && -0.00348 & -0.00385 & -0.00348 & -0.00272 && -0.01136 & -0.00486 & -0.00361 & -0.00255 \\
\hline 
  & & \multicolumn{4}{c}{ Uncontaminated  data}& & \multicolumn{4}{c}{Contaminated data}\\ 
 \cline{4-7} \cline{9-12}
$K_{is}=70$               &True value    && 0        & 0.2      & 0.4      & 0.6       && 0        & 0.2      & 0.4         & 0.6       \\ \hline 
$\eta_1$        & -1.38827 && -0.02188 & -0.01955 & -0.05972 & -0.10800 && 0.03391  & -0.01254 & -0.06567 & -0.12770 \\
$\eta_2$         & -0.48138 && -0.01318 & -0.01171 & -0.03481 & -0.06298 && 0.02199  & -0.00731 & -0.03817 & -0.07423 \\
$\eta_3$       & -6.46391 && -0.06923 & -0.06287 & -0.27869 & -0.55900 && 0.16868  & -0.03493 & -0.30595 & -0.67033 \\
$\alpha_1$      & 0.04946  && 0.00062  & 0.00055  & 0.00202  & 0.00446  && -0.00121 & 0.00033  & 0.00217  & 0.00530  \\
$\alpha_2$       & 0.04946  && 0.00054  & 0.00051  & 0.00235  & 0.00433  && -0.00128 & 0.00029  & 0.00260  & 0.00520  \\
$R(15,\boldsymbol{x}_0)$ & 0.91810  && -0.00174 & -0.00199 & -0.00127 & -0.00004 && -0.01041 & -0.00289 & -0.00121 & 0.00031  \\ 
\hline 
  & & \multicolumn{4}{c}{ Uncontaminated  data}& & \multicolumn{4}{c}{Contaminated data}\\ 
 \cline{4-7} \cline{9-12}
$K_{is}=100$              & True value    && 0        & 0.2      & 0.4      & 0.6       && 0        & 0.2      & 0.4         & 0.6       \\ \hline 
$\eta_1$        & -1.38827 && -0.01771 & -0.01652 & -0.06256 & -0.08467 && 0.04334  & -0.01518 & -0.06228 & -0.08025 \\
$\eta_2$        & -0.48138 && -0.01071 & -0.00996 & -0.03610 & -0.04904 && 0.02774  & -0.00875 & -0.03612 & -0.04659 \\
$\eta_3$        & -6.46391 && -0.05209 & -0.04718 & -0.27304 & -0.41072 && 0.21133  & -0.04194 & -0.27462 & -0.38078 \\
$\alpha_1$       & 0.04946  && 0.00057  & 0.00053  & 0.00204  & 0.00340  && -0.00145 & 0.00048  & 0.00226  & 0.00316  \\
$\alpha_2$      & 0.04946  && 0.00034  & 0.00030  & 0.00217  & 0.00315  && -0.00168 & 0.00026  & 0.00204  & 0.00295  \\
$R(15,\boldsymbol{x}_0)$ & 0.91810  && -0.00123 & -0.00140 & -0.00054 & 0.00003  && -0.01060 & -0.00223 & -0.00066 & -0.00012\\  \hline 
\end{tabular}
\end{table}

\begin{table}[p]\setlength{\tabcolsep}{2.1pt}  \renewcommand{\arraystretch}{1.7}
\caption{Bias for the semi-parametric model with $b=0.5$ and $c_0=6.5$. \label{table:PH_6500} }
\centering
\small
\begin{tabular}{r r c rrrr c rrrr}
     &          &&          &          &          &           &&          &          &             &           \\ \hline
  & & \multicolumn{4}{c}{ Uncontaminated  data}& & \multicolumn{4}{c}{Contaminated data}\\ 
 \cline{4-7} \cline{9-12}

$K_{is}=50$             & True value    && 0        & 0.2      & 0.4      & 0.6       && 0        & 0.2      & 0.4         & 0.6       \\ \hline 
$\eta_1$         & -0.66879 && 0.00223  & 0.00097  & 0.00020  & -0.00900 && 0.17046  & 0.14959  & 0.12180  & 0.11436  \\
$\eta_2$          & -0.01553 && 0.00320  & 0.00227  & 0.00156  & -0.00465 && 0.11845  & 0.10401  & 0.08415  & 0.08029  \\
$\eta_3$         & -4.42056 && -0.01196 & -0.00948 & -0.01543 & -0.03474 && 0.56481  & 0.50586  & 0.43705  & 0.36957  \\
$\alpha_1$     & 0.03000  && 0.00004  & 0.00001  & 0.00004  & 0.00021  && -0.00437 & -0.00389 & -0.00336 & -0.00286 \\
$\alpha_2$      & 0.03000  && 0.00014  & 0.00013  & 0.00018  & 0.00030  && -0.00426 & -0.00379 & -0.00324 & -0.00274 \\
$R(15,\boldsymbol{x}_0)$ & 0.87247  && -0.00529 & -0.00555 & -0.00534 & -0.00530 && -0.03609 & -0.03299 & -0.02952 & -0.02561 \\
\hline 
  & & \multicolumn{4}{c}{ Uncontaminated  data}& & \multicolumn{4}{c}{Contaminated data}\\ 
 \cline{4-7} \cline{9-12}
$K_{is}=70$               &True value    && 0        & 0.2      & 0.4      & 0.6       && 0        & 0.2      & 0.4         & 0.6       \\ \hline 
$\eta_1$         & -0.66879 && 0.00050  & -0.00537 & -0.00598 & -0.00468 && 0.16516  & 0.14230  & 0.12082  & 0.10748  \\
$\eta_2$         & -0.01553 && 0.00166  & -0.00294 & -0.00333 & -0.00217 && 0.11371  & 0.09772  & 0.08279  & 0.07497  \\
$\eta_3$         & -4.42056 && -0.03260 & -0.03492 & -0.03697 & -0.04026 && 0.55639  & 0.49446  & 0.42414  & 0.36377  \\
$\alpha_1$     & 0.03000  && 0.00016  & 0.00017  & 0.00019  & 0.00021  && -0.00434 & -0.00386 & -0.00329 & -0.00322 \\
$\alpha_2$    & 0.03000  && 0.00029  & 0.00032  & 0.00034  & 0.00036  && -0.00418 & -0.00367 & -0.00314 & -0.00318 \\
$R(15,\boldsymbol{x}_0)$& 0.87247  && -0.00224 & -0.00231 & -0.00226 & -0.00211 && -0.03358 & -0.03044 & -0.02652 & -0.02275 \\ 
\hline 
  & & \multicolumn{4}{c}{ Uncontaminated  data}& & \multicolumn{4}{c}{Contaminated data}\\ 
 \cline{4-7} \cline{9-12}
$K_{is}=100$              & True value    && 0        & 0.2      & 0.4      & 0.6       && 0        & 0.2      & 0.4         & 0.6       \\ \hline 
$\eta_1$         & -0.66879 && -0.00302 & -0.00428 & -0.00423 & -0.00453 && 0.15788  & 0.14182  & 0.12319  & 0.10443  \\
$\eta_2$          & -0.01553 && -0.00136 & -0.00242 & -0.00237 & -0.00256 && 0.10788  & 0.09716  & 0.08443  & 0.07154  \\
$\eta_3$         & -4.42056 && -0.02671 & -0.02485 & -0.02528 & -0.02720 && 0.55329  & 0.49534  & 0.43019  & 0.36647  \\
$\alpha_1$      & 0.03000  && 0.00019  & 0.00018  & 0.00019  & 0.00020  && -0.00422 & -0.00376 & -0.00325 & -0.00276 \\
$\alpha_2$     & 0.03000  && 0.00014  & 0.00014  & 0.00014  & 0.00015  && -0.00427 & -0.00381 & -0.00331 & -0.00281 \\
$R(15,\boldsymbol{x}_0)$ & 0.87247  && -0.00126 & -0.00139 & -0.00140 & -0.00135 && -0.03191 & -0.02878 & -0.02528 & -0.02190\\  \hline 
\end{tabular}
\end{table}

\begin{table}[p]\setlength{\tabcolsep}{2.1pt}  \renewcommand{\arraystretch}{1.7}
\caption{Bias for the semi-parametric model with $b=0.5$ and $c_0=6.5$.  \label{table:PH_6505}}
\centering
\small
\begin{tabular}{r r c rrrr c rrrr}
     &          &&          &          &          &           &&          &          &             &           \\ \hline
  & & \multicolumn{4}{c}{ Uncontaminated data}& & \multicolumn{4}{c}{Contaminated data}\\ 
 \cline{4-7} \cline{9-12}

$K_{is}=50$             & True value    && 0        & 0.2      & 0.4      & 0.6       && 0        & 0.2      & 0.4         & 0.6       \\ \hline 
$\eta_1$       & -1.38845 && -0.00292 & -0.02634 & -0.06961 & -0.10441 && 0.28565  & 0.19158  & 0.13289    & 0.07420   \\
$\eta_2$         & -0.48171 && -0.00007 & -0.01454 & -0.03906 & -0.08819 && 0.18184  & 0.12322  & 0.12394    & 0.12467   \\
$\eta_3$       & -7.28827 && -0.08460 & -0.15057 & -0.34018 & -0.97897 && 1.21433  & 0.81382  & 0.14850    & 0.33555   \\
$\alpha_1$       & 0.04946  && 0.00058  & 0.00102  & 0.00237  & -0.08206 && -0.00889 & -0.00592 & -0.00116   & -0.32679  \\
$\alpha_2$      & 0.04946  && 0.00057  & 0.00108  & 0.00246  & -0.10152 && -0.00891 & -0.00594 & -0.00112   & -0.39906  \\
$R(15,\boldsymbol{x}_0)$ & 0.96322  && -0.00163 & -0.00176 & -0.00145 & 0.00157  && -0.02785 & -0.01976 & -0.01062   & -0.00149  \\
\hline 
  & & \multicolumn{4}{c}{ Uncontaminated  data}& & \multicolumn{4}{c}{Contaminated data}\\ 
 \cline{4-7} \cline{9-12}
$K_{is}=70$               &True value    && 0        & 0.2      & 0.4      & 0.6       && 0        & 0.2      & 0.4         & 0.6       \\ \hline 
$\eta_1$        & -1.38845 && -0.01510 & -0.03564 & -0.05629 & -0.08443 && 0.28682  & 0.18935  & 0.03294     & 0.02196   \\
$\eta_2$        & -0.48171 && -0.00878 & -0.02069 & -0.03247 & -0.11914 && 0.18302  & 0.12036  & 0.02563     & 0.04542   \\
$\eta_3$        & -7.28827 && -0.06124 & -0.15643 & -0.22439 & -0.99978 && 1.24650  & 0.87564  & 0.14490  & -0.01407   \\
$\alpha_1$       & 0.04946  && 0.00041  & 0.00110  & 0.00155  & -0.01596 && -0.00916 & -0.00637 & -0.00104  & -0.15771   \\
$\alpha_2$      & 0.04946  && 0.00047  & 0.00111  & 0.00168  & -0.02075 && -0.00910 & -0.00631 & -0.00113  & -0.19334   \\
$R(15,\boldsymbol{x}_0)$ & 0.96322  && -0.00119 & -0.00105 & -0.00084 & 0.00146  && -0.02804 & -0.01977 & -0.00976  & -0.00119   \\ 
\hline 
  & & \multicolumn{4}{c}{ Uncontaminated  data}& & \multicolumn{4}{c}{Contaminated data}\\ 
 \cline{4-7} \cline{9-12}
$K_{is}=100$              & True value    && 0        & 0.2      & 0.4      & 0.6       && 0        & 0.2      & 0.4         & 0.6       \\ \hline 
$\eta_1$         & -1.38845 && -0.00904 & -0.01105 & -0.05924 & -0.23063  && 0.28616  & 0.19928  & 0.05644  &0.03762 \\
$\eta_2$         & -0.48171 && -0.00531 & -0.00635 & -0.03368 & -0.12308  && 0.18172  & 0.12619  & 0.03911 &-0.01015 \\
$\eta_3$    & -7.28827 && -0.06888 & -0.07584 & -0.29436 & -0.96654  && 1.22401  & 0.87425  & 0.21922 &-0.38512 \\
$\alpha_1$       & 0.04946  && 0.00048  & 0.00053  & 0.00202  & -0.00879  && -0.00897 & -0.00635 & -0.00168 &-0.09403 \\
$\alpha_2$      & 0.04946  && 0.00041  & 0.00047  & 0.00207  & -0.01198  && -0.00904 & -0.00642 & -0.00164 &-0.11596 \\
$R(15,\boldsymbol{x}_0)$ & 0.96322  && -0.00015 & -0.00021 & 0.00022  & 0.00236   && -0.02602 & -0.01814 & -0.00878 &-0.00050  \\  \hline 
\end{tabular}
\end{table}

\begin{figure}[p]
\centering
\begin{tabular}{cc}
\includegraphics[scale=0.4]{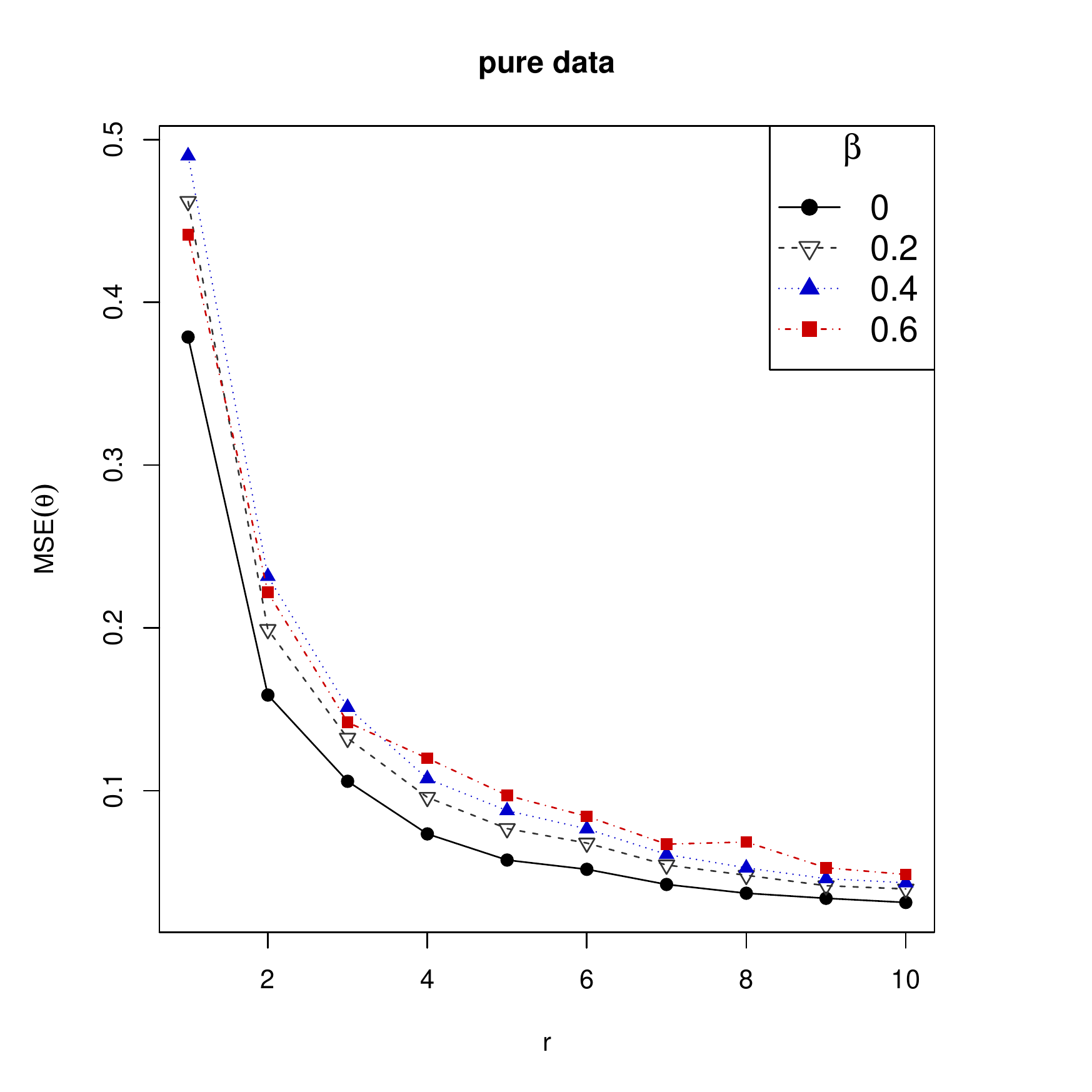} & 
\includegraphics[scale=0.4]{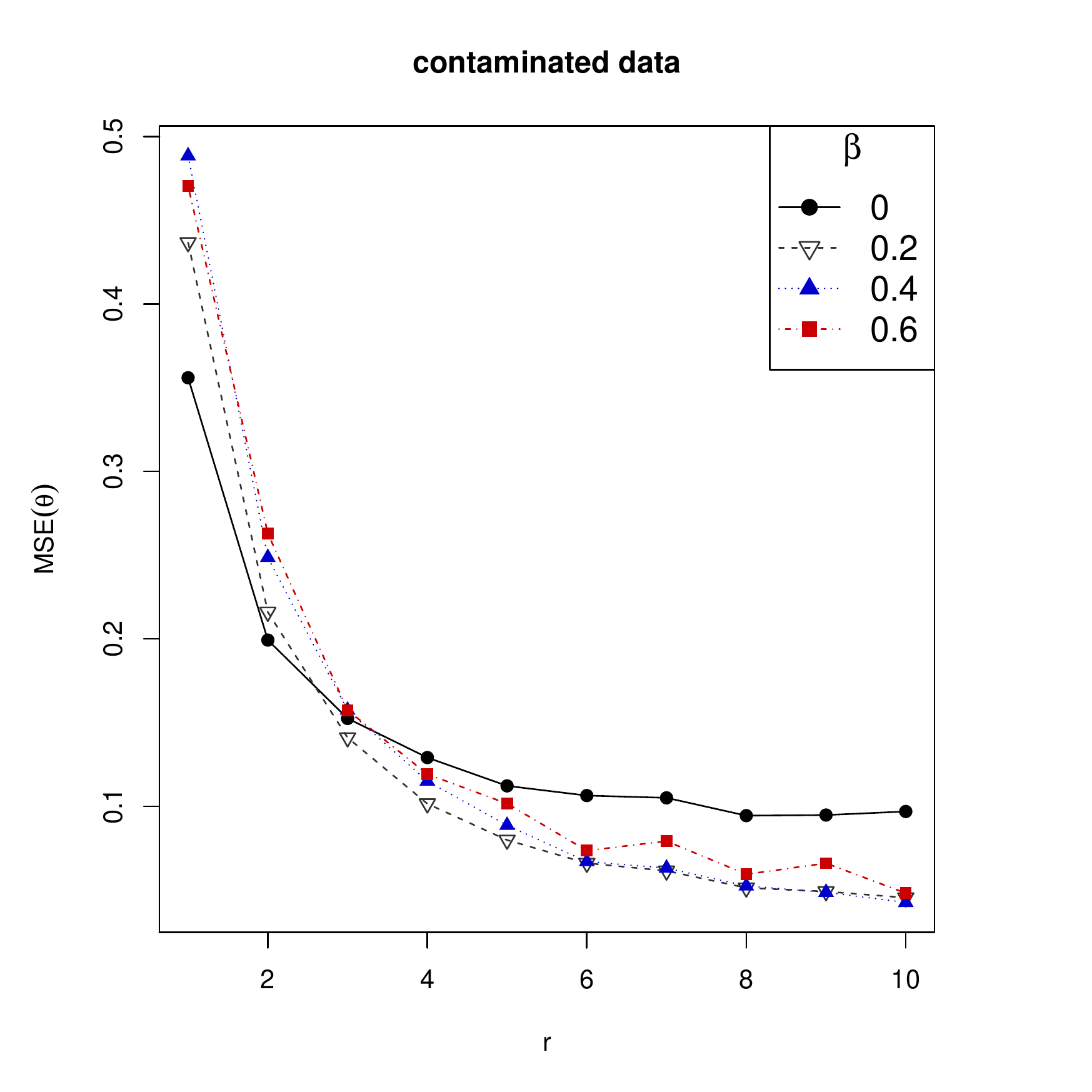} \\ 
\end{tabular}
\caption{MSEs for  unbalanced data \label{fig:PH_unbalanced}}
\end{figure}

\begin{figure}[p]
\centering
\begin{tabular}{cc}
\includegraphics[scale=0.4]{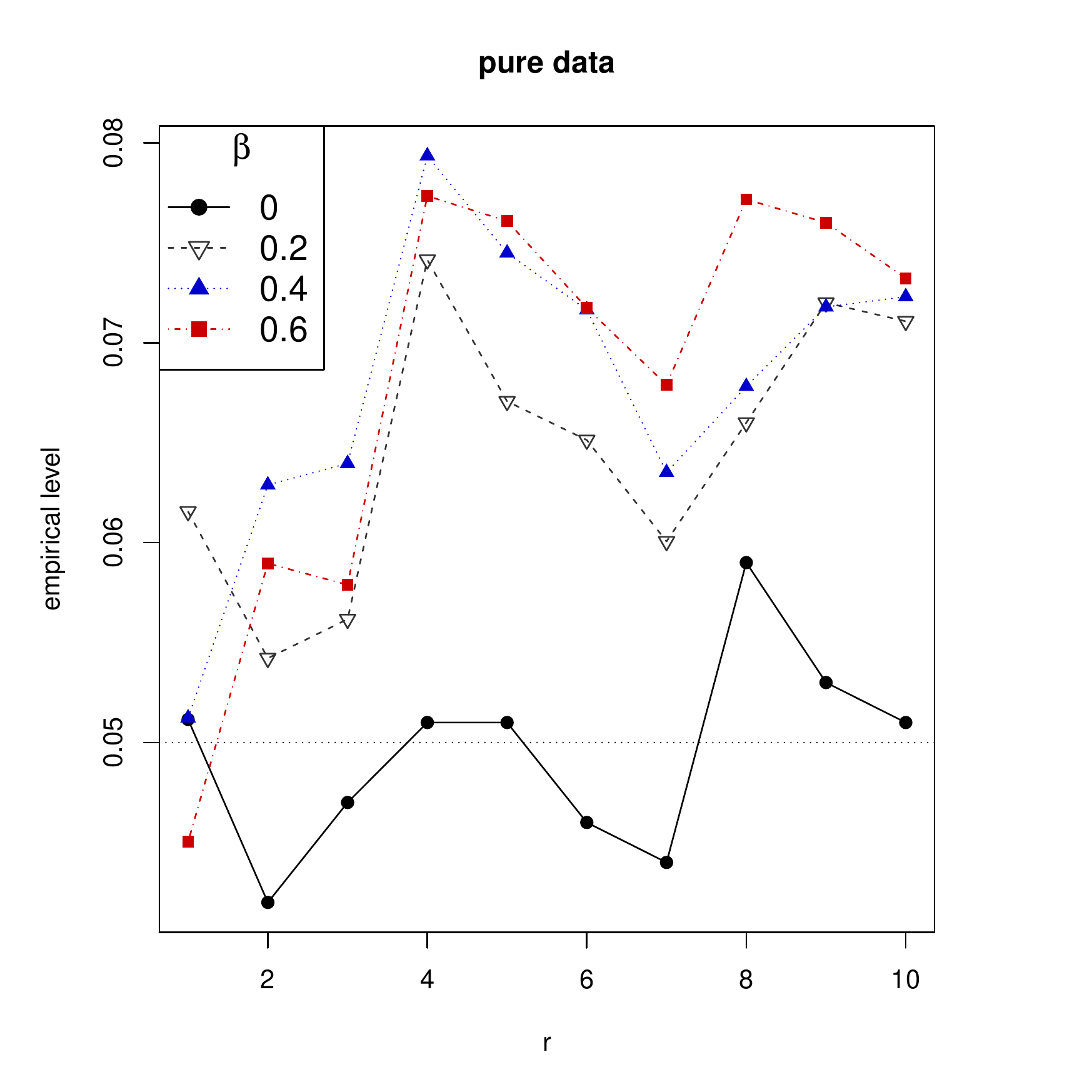} & 
\includegraphics[scale=0.4]{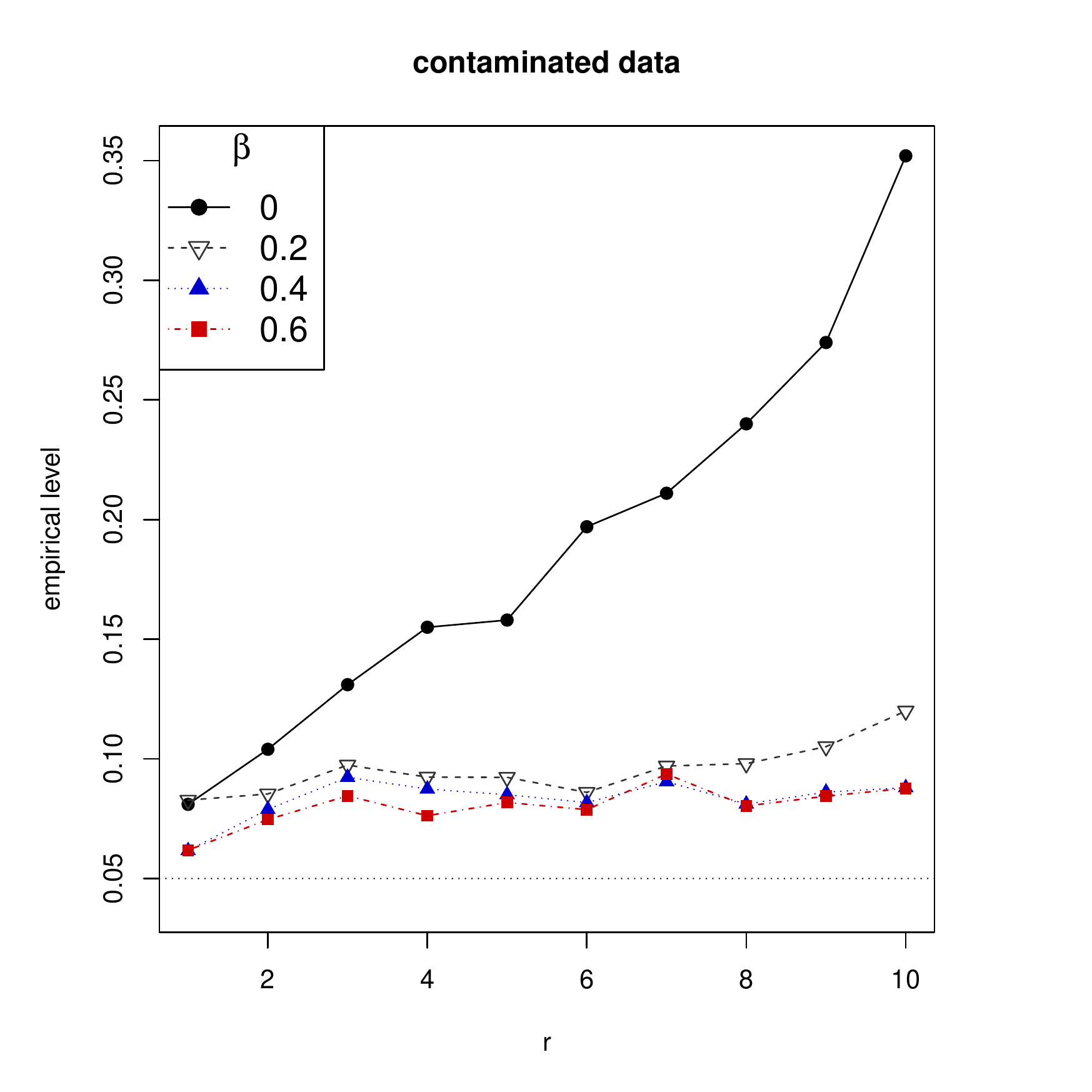} \\ 
\includegraphics[scale=0.4]{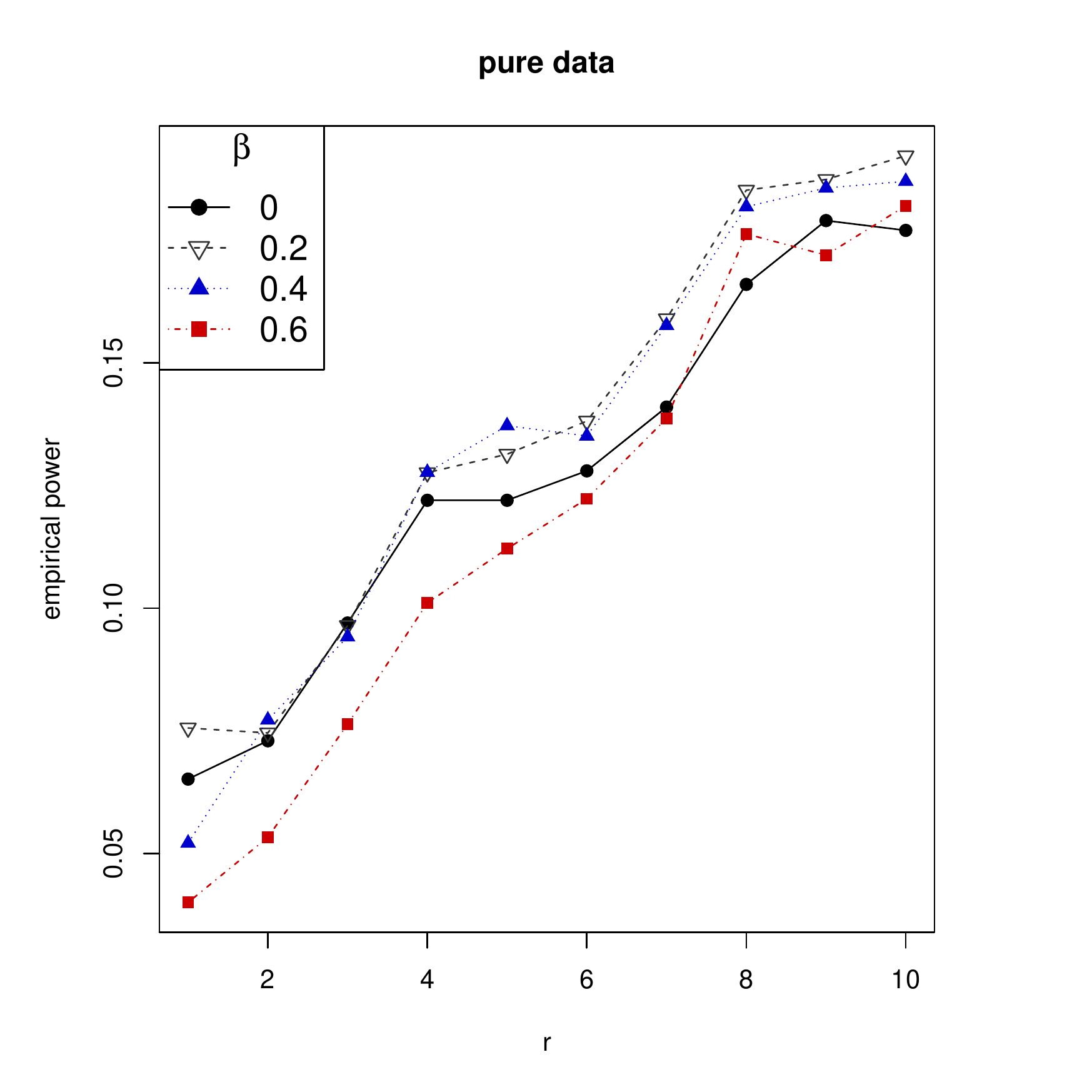} &
\includegraphics[scale=0.4]{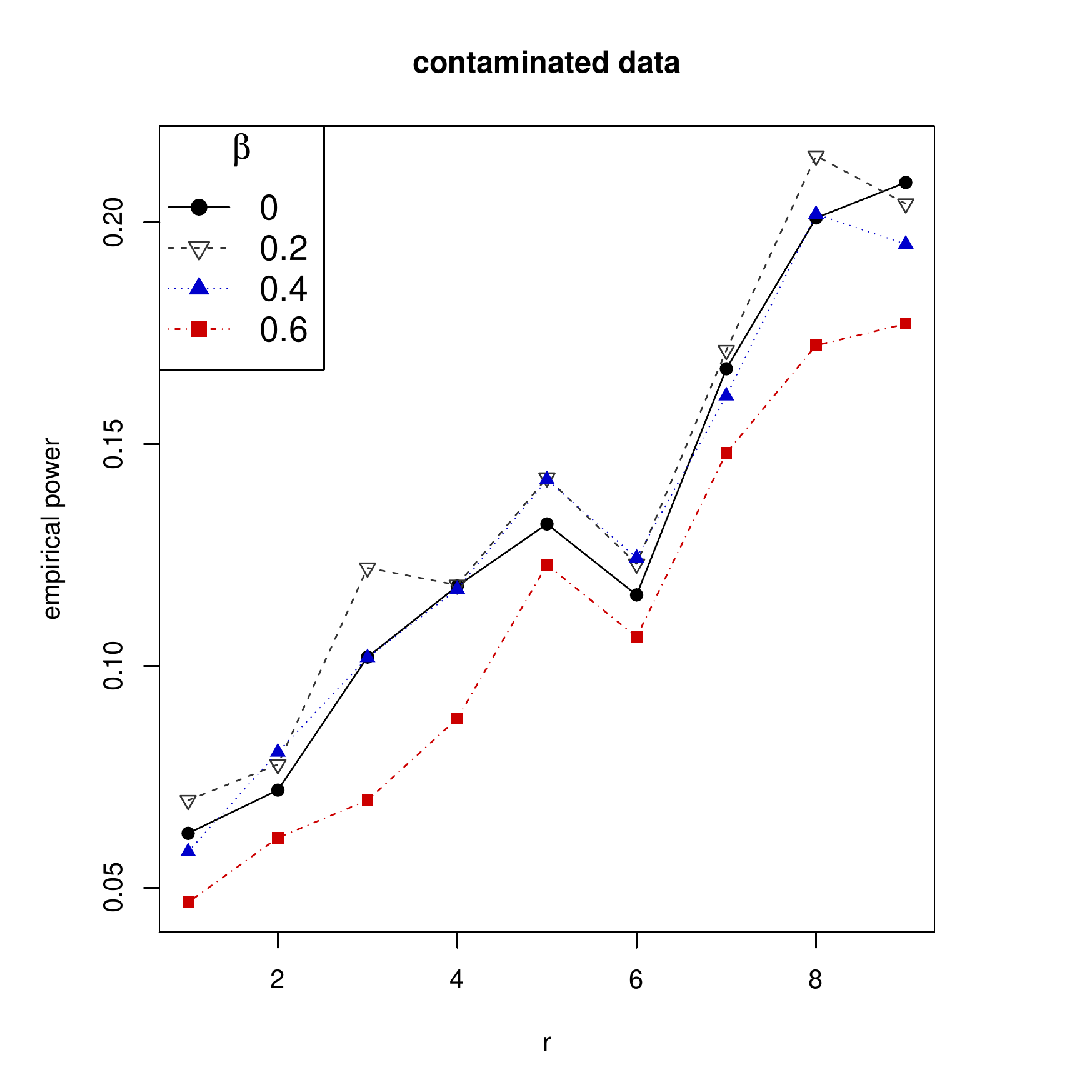} 
\end{tabular}
\caption{Estimated levels and powers for  unbalanced data \label{fig:PH_unbalanced_Wald}}
\end{figure}

\clearpage
\appendix

\section{Power function of  Wald-type tests}

In many cases, the power function of the proposed test procedure cannot be derived explicitly. In the following theorem, we present an useful asymptotic result for approximating the power function of the  Wald-type test statistics given in  (\ref{eq:PH_reject}). We shall assume that $\boldsymbol{\theta}^{\ast}\notin \boldsymbol{\Theta }_{0}$ is the true value of the parameter such that
\begin{equation*}
\label{eq:[15]}
\widehat{\boldsymbol{\theta}}_{\beta }\underset{K\rightarrow \infty }{\overset{P}{\longrightarrow }}\boldsymbol{\theta}^{\ast},
\end{equation*}%
and we denote 
$\ell _{\beta }\left( \boldsymbol{\theta}_{1},\boldsymbol{\theta}_{2}\right) =\boldsymbol{m}^{T}\left( \boldsymbol{\theta}_{1}\right) \left( \boldsymbol{M}^{T}\left( \boldsymbol{\theta}_{2}\right) \boldsymbol{\Sigma }_{\beta }\left( \boldsymbol{\theta}_{2}\right) \boldsymbol{M}\left( \boldsymbol{\theta}_{2}\right)\right) ^{-1}\boldsymbol{m}\left( \boldsymbol{\theta}_{1}\right) .$
We then have the following result.

\begin{theorem}
\label{res:wald2} We have%
\begin{equation*}
\sqrt{K}\left( \ell _{\beta }\left( \widehat{\boldsymbol{\theta}}_{\beta},\boldsymbol{\theta}^{*}\right) -\ell _{\beta }\left( \boldsymbol{\theta}^{\ast},\boldsymbol{\theta}^{\ast}\right) \right) \underset{K\rightarrow \infty }{\overset{\mathcal{L}}{\longrightarrow }}\mathcal{N}(0,\sigma _{W_{K},\beta}^{2}\left( \boldsymbol{\theta}^{\ast})\right),
\end{equation*}%
where 
\begin{equation*}
\sigma _{W_{K},\beta }^{2}\left( \boldsymbol{\theta}^{\ast}\right) =\left. \frac{\partial \ell _{\beta }\left( \boldsymbol{\theta},\boldsymbol{\theta}^{\ast}\right) }{\partial \boldsymbol{\theta}^{T}}\right\vert _{\boldsymbol{\theta}=\boldsymbol{\theta}^{\ast}}\boldsymbol{\Sigma }_{\beta }\left( \boldsymbol{\theta}^{\ast}\right)\left. \frac{\partial \ell _{\beta }\left( \boldsymbol{\theta},\boldsymbol{\theta}^{\ast}\right) }{\partial \boldsymbol{\theta}}\right\vert _{\boldsymbol{\theta}=\boldsymbol{\theta}^{\ast}}.
\end{equation*}
\end{theorem}
\begin{proof}
Under the assumption that 
\begin{equation*}
\widehat{\boldsymbol{\theta}}_{\beta }\underset{K\rightarrow \infty }{\overset{P}{\longrightarrow }}\boldsymbol{\theta}^{\ast},
\end{equation*}%
the asymptotic distribution of $\ell _{\beta }\left( \widehat{\boldsymbol{\theta}}_{1},\widehat{\boldsymbol{\theta}}_{2}\right) $ coincides with the asymptotic
distribution of $\ell _{\beta }\left( \widehat{\boldsymbol{\theta}}_{1},\boldsymbol{\theta}^{\ast}\right) .$ A first-order Taylor expansion of $\ell_{\beta }\left( \widehat{\boldsymbol{\theta}}_{\beta },\boldsymbol{\theta}\right) $ at $\widehat{\boldsymbol{\theta}}_{\beta }$, around $\boldsymbol{\theta}^{\ast}$, yields

\begin{equation*}
\left( \ell _{\beta }\left( \widehat{\boldsymbol{\theta}}_{\beta },\boldsymbol{\theta}^{\ast}\right) -\ell _{\beta }\left( \boldsymbol{\theta}^{\ast},\boldsymbol{\theta}^{\ast}\right) \right) =\left. \frac{\partial \ell _{\beta }\left( \boldsymbol{\theta},\boldsymbol{\theta}^{\ast}\right) }{\partial \boldsymbol{\theta}^{T}}\right\vert _{\boldsymbol{\theta}=\boldsymbol{\theta}^{\ast}}\left( \widehat{\boldsymbol{\theta}}_{\beta }-\boldsymbol{\theta}^{\ast}\right) +o_{p}(K^{-1/2}).
\end{equation*}%

Now, the result readily follows since 
\begin{equation*}
\sqrt{K}\left( \widehat{\boldsymbol{\theta}}_{\beta }-\boldsymbol{\theta}^{\ast}\right) \underset{K\rightarrow \infty }{\overset{\mathcal{L}}{\longrightarrow }}\mathcal{N}\left( \boldsymbol{0}_{J+1},\boldsymbol{\Sigma }_{\beta }\left( \boldsymbol{\theta}^{\ast}\right) \right) .
\end{equation*}
\end{proof}
\begin{remark}
Using Theorem \ref{res:wald2}, we can give an approximation for the power function of the Wald-type test statistic,  given in (\ref{eq:PH_reject}), at $\boldsymbol{\theta}^{\ast}$, as follows:
\begin{align*}
\pi _{W,K}\left( \boldsymbol{\theta}^{\ast}\right) & =\Pr \left( W_{K}\left( 
\widehat{\boldsymbol{\theta}}_{\beta }\right) >\chi _{r,\alpha }^{2}\right) \\
& =\Pr \left( K\left( \ell _{\beta }\left( \widehat{\boldsymbol{\theta}}_{\beta },\boldsymbol{\theta}^{\ast}\right) -\ell _{\beta }\left( \boldsymbol{\theta}^{\ast},\boldsymbol{\theta}^{\ast}\right) \right) >\chi _{r,\alpha }^{2}-K\ell _{\beta}\left( \boldsymbol{\theta}^{\ast},\boldsymbol{\theta}^{\ast}\right) \right) \\
& =\Pr \left( \frac{\sqrt{K}\left( \ell _{\beta }\left(\widehat{\boldsymbol{\theta}}_{\beta },\boldsymbol{\theta}^{\ast}\right) -\ell _{\beta }\left( \boldsymbol{\theta}^{\ast},\boldsymbol{\theta}^{\ast}\right) \right) }{\sigma _{W_{K},\beta}\left( \boldsymbol{\theta}^{\ast}\right) }>\frac{1}{\sigma _{W_{K},\beta}\left( \boldsymbol{\theta}^{\ast}\right) }\left( \frac{\chi _{r,\alpha }^{2}}{\sqrt{K}}-\sqrt{K}\ell _{\beta }\left( \boldsymbol{\theta}^{\ast},\boldsymbol{\theta}^{\ast}\right) \right) \right) \\
& =1-\Phi _{K}\left( \frac{1}{\sigma _{W_{K},\beta }\left( \boldsymbol{\theta}^{\ast}\right) }\left( \frac{\chi _{r,\alpha }^{2}}{\sqrt{K}}-\sqrt{K}\ell_{\beta }\left( \boldsymbol{\theta}^{\ast},\boldsymbol{\theta}^{\ast}\right) \right)\right)
\end{align*}%
for a sequence of distributions functions $\Phi _{K}\left( x\right) $ tending uniformly to the standard normal distribution $\Phi \left( x\right) $. It is clear that 
\begin{equation*}
\lim_{K\rightarrow \infty }\pi _{W,K}\left( \boldsymbol{\theta}^{\ast} \right) =1,
\end{equation*}%
i.e., the Wald-type test statistics are consistent in the sense of Fraser.
\end{remark}

\bibliography{bibliography} 
\bibliographystyle{abbrvnat}

\end{document}